\documentclass[mnsc,nonblindrev]{informs3} 

\OneAndAHalfSpacedXI 

\usepackage{natbib}
\bibpunct[, ]{(}{)}{,}{a}{}{,}%
\def\bibfont{\small}%
\usepackage{booktabs}

\usepackage{mathtools}
\usepackage{multirow}
\usepackage{mathtools, bm, bbm, color, enumerate, amsbsy, amsfonts, amsopn, amssymb,  amsxtra, bezier, graphicx, latexsym, verbatim,
	pictexwd, supertabular, url, dsfont,leftidx, setspace}
\usepackage[ruled,vlined,noend]{algorithm2e}
\newtheorem{eg}{Example}[section]
\newcommand{\mb}[1]{\ensuremath{\boldsymbol{#1}}}

\newcommand{\onee}{\mathbbm{1}}

\newcommand{\ntr}{{\tt In}}
\newcommand{\xtr}{{\tt Ex}}

\def\opt{\textsc{OPT}}
\def\alg{\textsc{GPG}}
\def\pg{\textsc{GPG}}
\def\aalg{\textsc{ALG}}

\DeclarePairedDelimiter{\ceil}{\lceil}{\rceil}

\SetAlFnt{\small}
\SetAlCapFnt{\small}
\SetAlCapNameFnt{\small}
\SetAlCapHSkip{0pt}
\IncMargin{-\parindent}

\def\falg{f\textsc{-GPG}} 

\def\mns{\text{-}}

\def\argmax{\text{argmax}}

\TheoremsNumberedThrough     
\ECRepeatTheorems
\EquationsNumberedThrough 

\begin{document}
\TITLE{Adwords with Unknown Budgets and Beyond} 
\ARTICLEAUTHORS{%
\AUTHOR{	Rajan Udwani}
\AFF{University of California Berkeley,  Industrial Engineering and Operations Research, 
\EMAIL{rudwani@berkeley.edu}}
}
\ABSTRACT{
In the classic Adwords problem introduced by Mehta et al.\ (2007), we have a bipartite graph between advertisers and queries. Each advertiser has a maximum budget that is known a priori. Queries are unknown a priori and arrive sequentially. When a query arrives, advertisers make bids and we (immediately and irrevocably) decide which (if any) Ad to display based on the bids and advertiser budgets. The winning advertiser for each query pays their bid up to their remaining budget. Our goal is to maximize total budget utilized without any foreknowledge of the arrival sequence (which could be adversarial). We consider the setting where the online algorithm does not know the advertisers' budgets a priori and the budget of an advertiser is revealed to the algorithm only when it is exceeded. A na\"ive greedy algorithm is 0.5 competitive for this setting and finding an algorithm with better performance remained an open problem. We show that no deterministic algorithm has competitive ratio better than 0.5 and give the first (randomized) algorithm with strictly better performance guarantee. We show that the competitive ratio of our algorithm is at least 0.522 but also strictly less than $(1-1/e)$. We present novel applications of budget oblivious algorithms in search ads and beyond. In particular, we show that our algorithm achieves the best possible performance guarantee for deterministic online matching in the presence of multi-channel traffic (Manshadi et al. (2022)).  
}
\maketitle

\section{Introduction}
Online advertising has emerged as the dominant marketing channel in many parts of the world. According to some estimates, in the year 2019, more than 450 billion USD were spent on online ads, which accounts for over 60\% of overall expenditure on ads~\footnote{Digital advertising spending worldwide from 2019 to 2024. https://www.statista.com/statistics/237974/online-advertising-spending-worldwide/}. Internet search is a prominent channel for online advertisement\footnote{In 2020 alone, Google made a revenue of over 104 billion USD from ``search \& other", which accounts for 70\% of their total revenue from advertising and exceeds half the total revenue of parent company Alphabet (see Alphabet Year in Review 2020. https://abc.xyz/investor/).}. 
In this medium, also called \emph{search ads}, advertisements are displayed alongside search results for key words relevant to the advertiser
Given a search request, advertisers make bids and then the platform 
decides which (if any) ads to show with the search results. 
The 
following model captures key elements of this problem.
\smallskip

\noindent \textbf{The Adwords problem~\citep{msvv}:}  At the beginning of the planning period (typically a day), the platform has a set $I$ of advertisers along with their budgets $(B_i)_{i\in I}$. Queries arrive sequentially on the platform and when a query $t$ arrives, advertisers make bids $(b_{i,t})_{i\in I}$. Given the bids and advertiser budgets, the platform immediately picks at most one advertisers' ad to display alongside the search results for the query\footnote{For simplicity, the model considers at most one ad slot per search result. A generalization to multiple ad slots is described in Section 6 of \cite{msvv}.}. The chosen advertiser pays their bid but only up to their remaining budget, i.e., at the end of the planning period, the total payment of an advertiser does not exceed his/her budget. The objective of the platform is to maximize the total advertiser budgets utilized without any foreknowledge of the arrival sequence (which could be adversarial). 
A typical modeling assumption is that the maximum individual bid by any advertiser is much smaller than the advertiser's initial budget. This 
is called the \emph{small bids} or \emph{large budgets} assumption and it is in line with the practice of search ads. 

\smallskip
%

In the face of uncertain future queries, which (if any) ad should the platform show for a given query? The greedy algorithm for this problem 
shows the ad with highest bid and (non-zero) available budget. In particular, if $I(t)$ is the set of advertisers which have non-zero remaining budget when query $t$ arrives, the greedy algorithm shows ad,
\[\underset{i\in I(t)}{\argmax}\quad  b_{i,t}.\]
This algorithm does not distinguish between two advertisers that make identical bids, even if the budget of one advertiser is nearly used up and the other has plenty of budget remaining. 
\cite{msvv} 
gave the state-of-the-art \emph{bid pricing} algorithm that 
carefully accounts for the fraction of remaining budgets. Given remaining budget $(B_i(t))_{i\in I}$ on arrival of query $t$, the algorithm of \cite{msvv} computes bid prices 
\[b_{i,t}\, (1-e^{-B_i(t)/B_i})\qquad \forall i\in I,\] 
and shows the ad with the largest bid price. Given two advertisers with similar bids, this algorithm may select the advertiser with smaller bid and higher fraction of remaining budget. The trade off between bid and (fraction of) remaining budget is governed by the function $(1-e^{-x})$ that arises quite naturally from the worst case performance analysis, a.k.a.\ \emph{competitive ratio} analysis, of this algorithm. The competitive ratio of an algorithm is the worst case relative performance gap between the online algorithm and (optimal) offline algorithm, \opt, that knows all the queries and advertiser bids. Let $G$ denote an instance of the problem (advertisers, arrivals, bids and budgets) and let $\mathcal{G}$ denote the set of all instances. Let $\textsc{ALG} (G)$ denote the expected total budget utilized by a (possibly) randomized online algorithm \textsc{ALG} on instance $G$. Similarly, let $\opt(G)$ denote the optimal offline value on instance $G$. 
\[\textbf{Competitive ratio of \textsc{ALG}:}\qquad  \min_{G\in \mathcal{G}}\frac{\textsc{ALG}(G)}{\opt(G)} \]
The algorithm of \cite{msvv} is $(1-1/e)$ competitive. In comparison, the greedy algorithm has a competitive ratio of 0.5. Remarkably, no (randomized) online algorithm has competitive ratio better than $(1-1/e)$ even for special cases of the Adwords problem~\citep{pruhs}.  

The knowledge of advertisers' initial budgets is essential for defining the bid pricing algorithm of \cite{msvv}. In fact, most (if not all) algorithms for Adwords in the literature (see \cite{devhay, alaei, mirrokni,survey, devsiv,lu}), rely critically on the knowledge of initial budgets for each advertiser. 
The na\"ive greedy algorithm is the one exception to this. For each query, greedy shows the ad with highest bid and (non-zero) available budget. Therefore, it is \emph{budget oblivious}, i.e., does not require any advance information about budgets except the knowledge of which advertisers are still participating (have non-zero remaining budget).  Clearly, a budget oblivious algorithm is more robust since it requires even less knowledge of the instance. 
Motivated by this, we consider the following question. 
\smallskip

\emph{
Is there a budget oblivious online algorithm for Adwords that outperforms greedy? What are the advantages of budget obliviousness? } 
\smallskip

We are not the first to ask these questions. In fact, the problem of Adwords with unknown budgets was first proposed by \cite{deb}, who used a special case of the problem to analyze 
a related setting called online matching with stochastic rewards (described in more detail in Section \ref{sec:stochrew}). Prompted by this intriguing connection, \cite{survey} posed the question of finding a budget oblivious algorithm that is better than greedy (or proving that no such algorithm exists) as an open problem (see Open Question 20 in \cite{survey}). To the best of our knowledge, this question remained open prior to our work.


\subsection{Our Contributions}\label{sec:ourc}
When budgets are unknown, we show that greedy 
is the best possible deterministic algorithm even when the budgets are large, i.e., no deterministic algorithm has competitive ratio better than 0.5 when budgets large but unknown. In fact, we show that this is true even when every bid is either 0 or 1, in which case the problem reduces to online matching (with large budgets). 


We show that one can do strictly better than greedy with randomized algorithms and give the first budget oblivious algorithm with competitive ratio better than 0.5. Our algorithm samples i.i.d.\ uniform random variables $x_i\in[0,1]\,\, \forall {i\in I}$, and computes random bid prices, 
{\color{black}\[b_{i,t}\, (1-e^{-\beta\, x_i})\qquad \forall i\in I,\] 
where $\beta>0$ is a parameter that can be chosen by the platform. 
At each query, the algorithm shows the ad with the largest (random) bid price from the set of advertisers with non-zero remaining budget. 
The key difference from the algorithm of \cite{msvv} is that bid prices are independent of budgets; the budget dependent factor, $(1-e^{-B_i(t)/B_i})$, in the algorithm of \cite{msvv}, 
is replaced with the random factor $(1-e^{-\beta\,x_i})$.} In the special case of Adwords where $b_{i,t}\in\{0,B_i\}\,\, \forall i\in I, t\in T$, our algorithm (with $\beta=1$) reduces to the 
Perturbed Greedy algorithm of \cite{goel}. Due to this connection, we refer to our algorithm as \emph{Generalized Perturbed Greedy} or GPG. 

Under the standard assumption of small bids, we show that for $\beta=1.15$, GPG is at least $0.522$ competitive against the optimal offline algorithm that knows all bids and budgets. For $\beta=1$, we show that the competitive ratio is least 0.508. Further, we prove that the competitive ratio of GPG is strictly less than $0.624$ ($<(1-1/e)$), for all $\beta\geq 1$. {\color{black}In fact, we numerically observe a stronger upper bound of 0.604 for $\beta=1$}\footnote{Very recently, \cite{zhihao} showed that for any function $f$, the algorithm that matches greedily using bid prices $b_{i,t} f(y_i)$ is strictly less than $(1-1/e)$ competitive.}. To prove the competitive ratio lower bound we overcome a novel obstacle that arises from combinatorial interactions between time varying bids and the randomness inherent in the algorithm. In particular, a ``dominance" property that is crucially used in most (if not all) previous analysis of (special cases) of GPG~\citep{goel, devanur, econ, albers, vazirani}, does not apply in the Adwords setting. 
We address this challenge by finding new structural insights into the problem. {\color{black} We also show that various parts of our analysis are individually tight but we believe that the overall analysis may not be tight.}

It is Adwords folklore that for small bids (large budgets) deterministic algorithms are as powerful as randomized ones, i.e., if there is a randomized $\alpha$ competitive algorithm then there exists a deterministic algorithm with competitive ratio at least $\alpha$.
Our results for unknown budgets imply a perhaps surprising gap between the two classes of algorithms. At a high level, 
randomness allows an online algorithm to (with some probability) conserve the budget of advertisers with low initial budgets 
without actually knowing the budgets. 

We demonstrate the usefulness of our budget oblivious algorithm in a variety of settings. First, we consider the setting of online matching (special case of Adwords) and multi-channel traffic. A stochastic generalization of this setting was introduced by \cite{multi} 
to model online recommendations on platforms such as VolunteerMatch.  
We show that our budget oblivious algorithm achieves the best possible competitive ratio guarantee for this setting, closing the gap between the lower and upper bound shown in \cite{multi}. Second, we consider a previously well known application of budget oblivious algorithms in the setting of online matching with stochastic rewards~\citep{deb}. Our result for Adwords with unknown budgets gives a new and more unified algorithmic result for this setting. Finally, we discuss potential applications of budget oblivious algorithms in auto-bidding and automated budget management, where a budget oblivious allocation algorithm may allow the platform to improve total budget utilization by re-optimizing advertisers' budgets during the planning period. 
\subsection{Related Work}\label{sec:ref}
\cite{kvv} introduced the classic setting of online bipartite matching that corresponds to Adwords with binary bids and unit budget for all advertisers. A bid of 1 denotes an edge in the bipartite graph and bid of 0 denotes the absence of an edge. Every advertiser can be matched to at most one neighboring query. \cite{kvv} showed (among other results) that randomly ranking advertisers at the start and then matching every query to the best ranked unmatched advertiser is a $(1-1/e)$ competitive algorithm for this setting. In fact, this algorithm, called Ranking, achieves the best possible guarantee for the problem.  
The analysis of Ranking was clarified and considerably simplified by \cite{baum} and \cite{goel2}. \cite{goel} considered the more general vertex weighted version of this problem 
which corresponds to Adwords with bids $b_{i,t}\in\{0,B_i\}\,\, \forall i\in I, t\in T$. They gave the Perturbed Greedy algorithm (that we use in this paper), and showed that the algorithm is $(1-1/e)$ competitive.

\cite{pruhs} considered the problem of online $b-$matching, which is a special case of Adwords with binary bids and identical (large) budget $b$ for every advertiser. They showed that as $b\to \infty$, the natural (deterministic) algorithm that balances the budget usage across advertisers is $(1-1/e)$ competitive. Generalizing this setting, \cite{msvv} introduced the Adwords problem 
and gave the bid pricing based $(1-1/e)$ algorithm for Adwords under the small bid assumption. \cite{buchbind} gave a primal-dual analysis for this algorithm. Subsequently, \cite{devanur} introduced the randomized primal-dual framework and used it to show all of the results mentioned above in a unified way. 

The Adwords setting without the small bids assumption, generalizes each of the settings discussed above. The budget-aware greedy algorithm that matches each query $t\in T$ according to the following rule, 
\[\argmax_{i\in I}\,\, (\min\{b_{i,t},B_i(t)\}),\]
where $B_i(t)$ is the remaining budget of $i$ on arrival of query $t$, is 0.5 competitive for Adwords without any assumption on the bids. 
\cite{kapralov} showed that without the small bids assumption, no online algorithm has competitive ratio better than 0.612 for Adwords\footnote{This applies when the competitive ratio is evaluated against an LP upper bound on the offline problem.}. Recently, \cite{huang2020adwords} gave the first algorithm with competitive ratio better than 0.5 for Adwords without the small bids assumption. 

\textbf{Concurrent with this paper:}  \cite{albers} and \cite{vazirani, Vazirani2022TowardsAP} independently show that the Perturbed Greedy algorithm is $(1-1/e)$ competitive for Adwords with binary bids and arbitrary resource budgets. 
To the best of our knowledge, these approaches do not yield a performance guarantee better than 0.5 for Generalized Perturbed Greedy (GPG) in the Adwords setting and this is posed as an open problem in \cite{Vazirani2022TowardsAP}. \cite{Vazirani2022TowardsAP} also identifies a key structural property, called \emph{no surpassing}, the absence of which prevents a generalization of their result to the Adwords setting. Through various numerical experiments (based on synthetic data), they demonstrate that the performance of GPG is at par with the algorithm of \cite{msvv}. We include a more detailed comparison with these papers in Appendix \ref{appx:related}. 

Very recently, \cite{zhihao} showed that for any function $f$, the algorithm that matches greedily using bid prices $b_{i,t} f(y_i)$, is strictly less than $(1-1/e)-\delta$ competitive for some constant $\delta>0$.

While the body of work discussed above considers an adversarial arrival sequence, there is also a long line of work on online matching and Adwords in stochastic and hybrid/mixed arrival models (for example, \cite{goel2,feldman,devhay,karande,vahideh,alaei,devsiv, mirrokni}).  For a comprehensive review of these settings we refer to \cite{survey}.

%
\medskip


\noindent \textbf{Outline for rest of the paper:} Section \ref{sec:prelim} discusses the assumption of small bids and presents a resource allocation version of the Adwords setting that we use interchangeably with the original formulation. In Section \ref{sec:deter}, we show that one cannot obtain a deterministic algorithm with competitive ratio better than 0.5. In Section \ref{sec:newmain}, we state and prove our main results for Adwords with unknown budgets. In Section \ref{sec:extension}, we show new results for our budget oblivious algorithm in settings beyond Adwords and discuss applications of budget obliviousness. Section \ref{sec:conclusion} concludes our discussion. 

%
%
%
%
%
%
\section{Preliminaries}\label{sec:prelim}
The Adwords problem generalizes a variety of settings in the literature on online resource allocation. We formalize this connection by discussing an online resource allocation problem that is known to be equivalent to the Adwords problem. Terminology from this setting and the Adwords problem will be used interchangeably throughout the paper.
\smallskip

\noindent \textbf{Online Budgeted Allocation (OBA):} Consider a complete bipartite graph $G$ with vertex set $(I,T)$. Vertices $i\in I$, called \emph{resources}, have capacities $(B_i)_{i\in I}$ and per-unit rewards $(r_i)_{i\in I}$. Resources, their rewards, and capacities are known to us. Vertices $t\in T$, called \emph{arrivals}, are unknown a priori and arrive sequentially. We use $T$ to denote the set of arrivals as well as the total number of arrivals. When a vertex $t\in T$ arrives, we see \emph{bids} $(b_{i,t})_{i\in I}$ 
that indicate $t$ is interested in up to $b_{i,t}$ amount of resource $i\in I$. Given the bids for arrival $t$, we must immediately and irrevocably match the arrival to at most one resource. If $t$ is matched to $i$, $b_{i,t}$ amount of $i$ are consumed, subject to availability.  
If the remaining capacity of $i$ is less than $b_{i,t},$ then matching $t$ to $i$ uses up all the remaining capacity. 
We receive a reward $r_i$ per-unit of $i$'s consumed capacity. The total reward from matching $i$ is capped at $r_iB_i$. The goal is to decide the allocation/matching for arrivals without any knowledge of future arrivals such that the total reward is maximized.
\smallskip

The Adwords problem is an instance of OBA where resources correspond to advertisers and arrivals correspond to queries. 
The per-unit rewards $r_i$ are set to 1 in Adwords for every $i\in I$. {\color{black} On the other hand, an instance of OBA with bids $b_{i,t}$, capacities $B_i$, and per-unit rewards $r_i$, is equivalent to the Adwords setting with scaled bids $r_ib_{i,t}$ and budgets $r_iB_i$.} Due to this observation, without loss of generality (w.l.o.g.), we let per-unit rewards  
\[r_i=1 \quad \forall i\in I.\]
A standard assumption in the Adwords setting is the that the maximum bid-to-budget ratio, \[\gamma:= \max_{i\in I,t\in T}\frac{b_{i,t}}{B_i}.\]
is small, i.e., 
$\gamma\to 0$. This is also called the \emph{small bid} or the \emph{large budget/capacity} assumption. This assumption is in line with 
the practice of search ads, where individual bids are typically much smaller than the overall budget. Recall that the $(1-1/e)$ guarantee of \cite{msvv} holds only in the small bid regime. Even in this regime, no online algorithm has a competitive ratio better than $(1-1/e)$. 

\noindent \textbf{OBA with Unknown Capacities:} In the setting we are interested in, the online algorithm has no prior knowledge of resource capacities. When the algorithm fully utilizes the capacity of a resource, the event is immediately revealed to the algorithm, i.e., capacity of a resource is revealed right after it is used up. 
We evaluate the competitive ratio of a capacity/budget oblivious online algorithm against the optimal offline matching with complete knowledge of the instance. 





\section{Upper Bound for Deterministic Algorithms}\label{sec:deter}
The folklore for Adwords and many other online resource allocation problems says that 
given an online randomized algorithm, one can construct an online  deterministic algorithm that matches arrivals \emph{fractionally} and emulates the expected performance 
of the randomized algorithm at every arrival. As we demonstrate below, this line of argument relies on prior knowledge of budgets.  

\begin{eg}\label{determeg}\emph{
Consider an instance of online matching with a \emph{big} resource that has capacity 2 and a \emph{small} resource that has capacity 1. The first two arrivals have an edge to both resources (bids of 1). The Ranking algorithm starts with a random ranking of the two resources and matches every arrival to an available resource with the best rank. The first arrival is matched uniformly randomly. The second arrival is matched to the big resource with probability (w.p.) 1; either Ranking matches both arrivals to the \emph{big} resource (w.p. 0.5) or matches the second arrival to the \emph{big} resource 
after discovering (w.p. 0.5) that the budget of \emph{small} resource has been used up. A deterministic algorithm that tries to emulate Ranking will match equal fractions (0.5) of the first arrival to both resources. Now, half of the \emph{small} resource's budget is available at the second arrival and without knowing the budgets, the deterministic algorithm fails to distinguish between the two resources. 
Given a third arrival that can only be matched to the small resource, Ranking matches it w.p. 0.5 but 
the deterministic counterpart will fail to save the \emph{small} resources' budget and cannot match any fraction of the final arrival. 
}\end{eg}
Building on this example, we establish the following result.
\begin{theorem}
Every deterministic budget oblivious online algorithm for Adwords has competitive ratio at most 0.5, even on instances with binary bids and large budgets. 
\end{theorem}
\begin{proof}{Proof.}
Let ALG denote a deterministic budget oblivious online algorithm. For a deterministic algorithm, we can construct a worst case arrival sequence and assign resource capacities based on the decisions made by the algorithm. We construct an instance with $n$ resources and binary bids. The arrival sequence has two phases. Given an arrival and a resource, an edge between the two corresponds to a bid of 1 and the absence of an edge corresponds to a bid of 0. Let $B>0$ be an arbitrary integer value. We assign capacity $nB$ to one resource and capacity $B$ to the other $n-1$ resources.

Phase one arrivals occur first and this phase has $nB$ arrivals that have an edge to every resource. When ALG matches a resource $B$ times, we reveal (to ALG) that the capacity of the resource is $B$, i.e., the resource has been used up. 
We do this up to $n-1$ times in total, i.e., when a resource $i$ is matched to $B$ arrivals and $i$ is not the sole resource with non-zero remaining capacity, we then inform the algorithm \alg\ that the capacity of $i$ has been exhausted. There is at least one resource with non-zero reamining capacity at the end of phase one. Let $j^*$ denote one such resource. 
Phase two has $(n-1)B$ arrivals that occur after phase one arrivals and have an edge to every resource except $j^*$. Therefore, ALG can only match these arrivals to resources in $[n]\backslash \{j^*\}$ with available capacity. The capacity of resource $j^*$ is $nB$ and the capacity of all other resources is $B$. 

Note that, optimal offline (\opt) matches every arrival in phase one to resource $j^*$ and matches phase two arrivals to the other resources, fully utilizing the total capacity. 
\[\opt= (2n-1)B.\]  
At the end of both phases, ALG matches at most $B$ arrivals to $j^*$ and therefore, at least $(n-1)B$ of the total resource capacity is unused in ALG. 
For $n\to+\infty$, we have
\[\frac{\text{ALG}}{\opt}\leq 1- \lim_{n\to +\infty} \frac{n-1}{2n-1}=0.5.\] 
Since the value of $B>0$ can be arbitrary, this upper bound also applies to instances with large budgets. 

\hfill\Halmos\end{proof}
\section{Randomized Budget Oblivious Algorithm and Analysis}\label{sec:main}\label{sec:randlb}\label{sec:newmain}
In this section, we state and prove our main results for the Adwords/OBA problem with unknown budgets. 
Consider the following family of randomized algorithms with parameter $\beta> 0$.

\begin{algorithm}[H]
\textbf{Inputs:} Set of advertisers $I$, parameter $\beta$\; 
Let $g(x)=e^{\beta(x-1)}$\;
For every $i\in I$ generate i.i.d.\ r.v.\ $y_i\sim U[0,1]$\;
\For{\text{every new arrival } $t$}{
Match $t$ to $i^*=\underset{ i\in I}{\arg\max}\,\, b_{i,t} (1-g(y_i))$\;
\textbf{if} {$i^*$ is out of budget} \textbf{then} \text{update} $I=I\backslash\{i^*\}$\;			}

\caption{Generalized Perturbed-Greedy (\alg)}
\label{rank}
\end{algorithm}

Observe that Algorithm \ref{rank} (\alg) is budget oblivious and matches every arrival $t\in T$ greedily based on \emph{randomized bid prices} $b_{i,t}\, (1-g(y_i))\,\, \forall i\in I$. The uniform random variables $(y_i)_{i\in I}$, called \emph{seeds}, are sampled independently for each $i\in I$. On any given problem instance, ties between bid prices of any two (different) resources occur with a probability of 0. 

The choice of function $g(\cdot)$ is influenced by the competitive ratio analysis of the algorithm. 
In particular, choosing $g(x)=e^{x-1}$ gives the optimal algorithm for closely related settings such as vertex weighted online matching~\citep{goel}. Motivated by this, we will focus on analyzing the competitive ratio of \alg\ for the family of exponential functions $g(x)=e^{\beta(x-1)}$ for $\beta>0$. {\color{black} Note that in Section \ref{sec:ourc}, we presented \alg\ with seed values $x_i=1-y_i$. For consistency with prior work~\citep{goel, devanur}, we use seeds $y_i$ in the subsequent discussion. }
{\color{black}
As we will discuss later, it is challenging to show a non-trivial guarantee for \alg\ directly. We lower bound the competitive ratio of \alg\ by first analyzing Algorithm \ref{frank} ($\falg$), which is a fractional relaxation of \alg.
\begin{algorithm}
	\textbf{Inputs:} Set of advertisers $I$, budgets $(B_i)_{i\in I}$, parameter $\beta$\;
	Initialize $I(0)=I$, $g(t)=e^{\beta(t-1)}$ and $B_i(0)=B_i$ for every $i\in I$\;
	Generate i.i.d.\ sample $y_i\in U[0,1]$ for every $i\in I$\;
	\For{\text{every new arrival } $t$}{
		Initialize total fractional of arrival matched $\delta=0$\; 
Initialize set of available advertisers $I(t)=I(t-1)$ and remaining budgets $B_i(t)=B_i(t-1)\,\, \forall i\in I$\; 
		\While{$\delta<1$ \textbf{\emph{and}} $I(t)\cap \{i\mid b_{i,t}>0\}\neq \emptyset$}{			
			Let $i^*=\underset{ i\in I(t)}{\arg\max}\,\, b_{i,t}r_i (1-g(y_i))$\; 
			Let $\delta_{i^*}(t)=\min\left\{1\,,\, \frac{B_{i^*}(t-1)}{b_{i^*,t}}\right\}$\;
			\textbf{if} $\delta_{i^*}(t)>0$ \textbf{then} update $B_{i^*}(t)=B_{i^*}(t-1)-b_{i^*,t}\,\delta_{i^*}(t)$ and $\delta=\delta+\delta_{i^*}(t)$\;
			Update the set of available advertisers $I(t)=\{i\mid B_i(t)>0\}$\;
	}
}
\caption{Fractional GPG ($\falg$)}
\label{frank}
\end{algorithm}
Clearly, $\falg$ is not budget oblivious. However, this is not of concern as the algorithm is only used as an intermediate step to analyze \alg.
$\falg$ is also a randomized algorithm where random seeds $(y_i)_{i\in I}$ are chosen at the beginning. Similar to \alg, ties between bid prices occur with a probability of 0 in $\falg$. The key difference in $\falg$ is that each arrival $t\in T$ is matched fractionally to (possibly) multiple resources. For $t\in T$ and some fixed values of the seeds $(y_i)_{i\in I}$, let $I(t)$ denote the set of resources with available budget when $t$ arrives and let $B_i(t)$ denote the available budget of resource $i$ right after arrival $t$ departs. In \alg, we simply match an arrival $t$ to resource $i^*$ that has non-zero budget available and has the highest randomized bid price, whereas, in $\falg$, if the remaining budget of $i^*$ is less than the bid $b_{i^*,t}$, i.e., $B_{i^*}(t-1)<b_{i^*,t}$, then we only match a fraction $\frac{B_{i^*}(t-1)}{b_{i^*,t}}$ of arrival $t$ to $i^*$. For example, if $B_{i^*}(t-1)=0.5b_{i^*,t}$, then we match 0.5 fraction of $t$ to $i^*$ and this uses up the remaining budget of $i^*$. Then, from the remaining set of resources we again select a resource with the highest (randomized) bid price and repeat the process, until $t$ is fully (fractionally) matched or no resource is available. The total amount of resource $i$'s budget allocated to arrival $t$ is given by $B_i(t)-B_i(t-1)$. Observe that an arrival is fractionally matched to at most $|I|$ different resources and at most $|I|$ arrivals are matched fractionally in $\falg$. 
 We show the following lower bound on the competitive ratio of $\falg$ against the optimal offline integral matching. 

\begin{theorem}\label{guafalg}
With $\beta=1.15$, $\falg$ is at least $0.522$ competitive against optimal (integer) offline allocation. When $\beta=1$, $\falg$ is at least $0.508$ competitive.
\end{theorem}
Given the similarities between $\falg$ and \alg\ it might be tempting to conclude similar competitive ratio results for \alg\ in the small bids regime. 
Perhaps surprisingly, it turns out that for the same seed values, \alg\ and $\falg$ can sometimes output very different matchings (see Example \ref{integfrac} in Appendix \ref{appx:challenge}). Nonetheless, we show that the objective value of the two algorithms is always similar and using this we establish the following guarantee for \alg. 
\begin{theorem}\label{main}
	Given a competitive ratio guarantee $\eta$ for $\falg$ (for some parameter value $\beta$), we have that \alg\ (with the same value of $\beta$) is $\frac{1}{1+\gamma}\,\eta$ competitive. 
\end{theorem}
\begin{corollary}\label{maincoro}
	For $\beta=1.15$, \alg\ is at least $\frac{1}{1+\gamma}0.522$ competitive. 
When $\beta=1$, \alg\ is at least $\frac{1}{1+\gamma}0.508$ competitive.
\end{corollary}
The corollary above follows by combining Theorems \ref{guafalg} and \ref{main}.
These results hold for any instance of OBA, i.e., for arbitrary value of $\gamma$. For small bids we have $\gamma\to0$ and 
the competitive ratio of \alg\ is strictly better than 0.5. 
 We also establish that the competitive ratio of \alg\ is strictly less than $(1-1/e)$. 
\begin{theorem}\label{gpgub}
	For every bid-to-budget ratio $\gamma\in(0,1]$ and $\beta\geq 1$, 
	the competitive ratio of \alg\ is at most $0.624$.
\end{theorem}
The proof of Theorem \ref{gpgub} is analytical but Monte Carlo simulations indicate that an even stronger upper bound of 0.604 may hold for $\beta=1$. For special cases of Adwords, our analysis yields a stronger guarantee that surpasses these upper bounds, as described in the next lemma\footnote{The $(1-1/e)$ guarantee for OBA with integer starting capacities and binary bids was also shown concurrently in \cite{albers} and \cite{Vazirani2022TowardsAP}. Further, the ``no-surpassing" condition of \cite{Vazirani2022TowardsAP} generalizes the notion of decomposable bids.}.
\begin{theorem}\label{decomp}
	For $\beta=1$, \alg\ is $\frac{1}{1+\gamma}(1-1/e)$ competitive for OBA when the bids are decomposable, i.e., $b_{i,t}\in\{0,b_i\times b_t\}\,\, \forall i\in I, t\in T$. For the special case of OBA with integer starting capacities and binary bids, \alg\ is $(1-1/e)$ competitive for arbitrary bid-to-budget ratio $\gamma$.
\end{theorem}

\noindent  In the following sections, we focus on 
proving Theorem \ref{guafalg} (Sections \ref{sec:overview} and \ref{sec:ana}) and Theorem \ref{main} (Section \ref{sec:fracint}). In the course of proving these results, we discuss the insights behind the improved guarantees stated in Theorem \ref{decomp} and defer a formal 
proof of Theorem \ref{decomp} to Appendix \ref{appx:decompose}. We discuss the main insights behind Theorem \ref{gpgub} in Section \ref{sec:randlb} and defer the formal proof of this result to Appendix \ref{appx:gpgub}.

\subsection{Overview of the Analysis of $\falg$}\label{sec:overview}
We start with some notation. Let the number of resources $|I|=n$. Let \opt\ refer to the (integer) optimal offline algorithm. Overloading notation, we also use \opt\ to denote the total reward of the optimal offline algorithm. Since there are no unknowns in the offline problem, i.e., budgets and bids are all known, the optimal offline solution is a deterministic  matching. Let $\opt_i$ denote the set of arrivals matched to $i$ as well as the total fraction of $i$'s budget that is matched in \opt. 
Note that $\opt=\sum_{i\in I}  \opt_i$.  

Because the discussion prior to Section \ref{sec:mainchal} applies to both \alg\ and $\falg$,
we use \aalg\ as a unified reference in place of separate references to \alg\ and to $\falg$. Overloading notation, we also use \aalg\ to refer to the expected reward of the online algorithms. Let $Y$ denote the vector of seeds $(y_i)_{i\in I}\in[0,1]^n$. For fixed seeds $Y,$ \aalg\ is deterministic, i.e., both \alg\ and $\falg$ are deterministic.   For this reason, we also refer to the vector $Y$ as a sample path of \aalg. 
Let $\aalg(Y)$ denote the (fractional) matching generated by \aalg\ on sample path $Y$. Let $\aalg_i(Y)$ denote the set of arrivals (fractionally) matched to $i$ in the matching $\aalg(Y)$.  Overloading notation, we also use $\aalg(Y)$ to denote the total reward of \aalg\ with seed $Y$ and $\aalg_i(Y)$ to denote the total budget of $i$ used in $\aalg(Y)$. Let $\aalg_t(Y)$ denote the set of resources (fractionally) matched to $t$ in $\aalg(Y)$ and let $\delta_{i,t}(Y)$ denote the fraction of $t$ matched to $i$ in $\aalg(Y)$. Observe that in \alg, the set $\alg_t(Y)$ will never have more than one element. 

 The (randomized) primal-dual method of \cite{devanur} is a standard technique for analyzing online algorithms for bipartite matching. Unfortunately, it is not obvious to us if one can use this method to obtain a non-trivial lower bound on the competitive ratio of \aalg. We illustrate this in more detail 
in Appendix \ref{appx:pmd}. To prove Theorem \ref{guafalg}, we use the flexible LP free analysis framework of \cite{full}. To prove a lower bound on the competitive ratio of an online algorithm ALG in this framework, it suffices to find a feasible solution to the following system of (linear) inequalities 
in variables $\lambda_t$ and $\theta_i$, 
\begin{eqnarray}
	\sum_{t\in T} \lambda_t + \sum_{i\in I} \theta_i &\leq & (1+\varepsilon)\, \text{ALG} \label{dual1}\\
	\theta_i+\sum_{t\in \opt_i} \lambda_t  &\geq &\alpha\, \opt_i \qquad \forall i\in I, \label{dual2}\\
	\lambda_t\geq 0, &&\theta_i\geq 0\qquad \forall t\in T,\, i\in I. \label{dual3}
\end{eqnarray}
\begin{lemma}\label{lpframe}
	Given a solution to the system defined by \eqref{dual1}$-$\eqref{dual3}, we have that \emph{ALG} is $\frac{\alpha}{1+\varepsilon}$ competitive against \opt.
\end{lemma}
\begin{proof}{Proof.}
	Summing up inequalities \eqref{dual2} over all $i\in I$, we have
	\begin{eqnarray*}
		\alpha\, \opt \,=\,	\alpha\, \sum_{i\in I} \opt_i \leq \sum_{i\in I} \sum_{t\in \opt_i} \lambda_t + \sum_{i\in I}\theta_i\, \overset{(*)}{\leq}\,
		\sum_{t\in T}\lambda_t + \sum_{i\in I}\theta_i\, \leq\, (1+\varepsilon)\, \text{ALG},
	\end{eqnarray*}
	here inequality $(*)$ follows from the fact that the optimal offline solution is integral, i.e., \opt\ matches each arrival to at most one resource and $\opt_i\cap \opt_j=\emptyset$ for any two (distinct) resources $i,j\in I$.
	\hfill\Halmos\end{proof}

Due to Lemma \ref{lpframe}, our main goal is to find a feasible solution to the system \eqref{dual1}$-$\eqref{dual3} with a suitably large value of $\frac{\alpha}{1+\varepsilon}$. Recall that $\aalg_t(Y)$ is the set of resources matched to $t$ in $\aalg(Y)$. To define the candidate solution, let
\begin{eqnarray}
	&& \lambda_t(Y)= 
\sum_{j\in \aalg_t(Y)} b_{j,t}\,\delta_{j,t}(Y)(1-g(y_j))\qquad  \forall t\in T,\, Y\in [0,1]^n,\label{lambda1}\\
	&& \theta_i(Y) =\sum_{t\in \aalg_i(Y)} b_{i,t}\,\delta_{i,t}(Y)\, g(y_i)\,=\,\aalg_i(Y)\, g(y_i)\qquad \forall i\in I,\, Y\in [0,1]^n.\label{theta1}
\end{eqnarray}
Our candidate solution is,
	\begin{eqnarray}
		&&\lambda_t = E_Y\left[\lambda_\tau(Y)\right]\quad \forall t\in T  \text{ and  }\quad
		\theta_i=E_Y[\theta_i(Y)]\quad  \forall i\in I,
		\label{dualdef}
	\end{eqnarray}
here $E_Y[\cdot]$ denotes expectation over the random seeds in \aalg.	The main idea behind this candidate solution is as follows. Matching a $\delta_{i,t}(Y)$ fraction of $t$ to $i$, generates reward $b_{i,t}\delta_{i,t}(Y)$. We split this reward is into two parts; $(1-g(y_i))$ fraction of the reward is added to $\lambda_t(Y)$ and the remaining $g(y_i)$ fraction of the reward is added to $\theta_i(Y)$. Then, we take expectation over the seeds $Y$. This reward splitting is a natural generalization of the one used \cite{devanur} for vertex weighted online bipartite matching. 
	
	The main challenge is to prove that this candidate solution satisfies inequalities \eqref{dual2} for a large enough value of $\alpha$. At a high level, we follow the same strategy as \cite{devanur}. Since it is challenging to analyze the expectation of multi-variate random variables such as $\lambda_t(Y)$ and $\theta_i(Y)$, 
	we fix the seeds for all but one resource, say $i\in I$, and evaluate conditional expectations with respect to (w.r.t.) randomness in $y_i$. Let $Y_{-i}\in[0,1]^{n-1}$ denote the vector of seeds of all resources except $i$. We fix $Y_{-i}$ and establish non-trivial lower bounds on the conditional expectations $E_{y_i}[\lambda_t(y_i, Y_{-i})\mid Y_{-i}]$ and $E_{y_i}[\theta_i(y_i, Y_{-i})\mid Y_{-i}]$. Then, combining the lower bounds gives us \eqref{dual2}.
	
	To lower bound the conditional expectations, we examine the changes in the matching generated by \aalg\ for different values of $y_i$. In particular, we compare the matching for any given value $y_i<1$, with the matching for  
	$y_i=1$. The latter scenario serves as a base scenario where the bid price of $i$ is 0 everywhere (since $1-g(1)=0$). 
	To lower bound  $E_{y_i}[\lambda_t(y_i,Y_{-i})\mid Y_{-i}]$, we show that $\lambda_t(y_i,Y_{-i})$ takes its minimum value in the base scenario, i.e., when $y_i=1$. This is similar to the ``monotonicity" property shown in \cite{devanur}. 
	To obtain a sufficiently strong lower bound on $E_{y_i}[\theta_i(y_i,Y_{-i})\mid Y_{-i}]$, ideally, we hope that as $y_i$ decreases from 1 to 0, the amount of $i$'s budget used ($\aalg_i(y_i,Y_{-i})$) increases monotonically at a ``fast enough" rate. This corresponds to the ``dominance" property in \cite{devanur}. All known analysis of (special cases) of the Perturbed Greedy algorithm reply on establishing these two properties \citep{goel, devanur, stochrew, econ, vazirani, albers, delong}. 
	
	\subsubsection{Main Challenge.}\label{sec:mainchal} Unfortunately, as we illustrate in Example \ref{ex} below, the ``dominance" property does not hold for \alg. In fact, $\alg_i(y_i,Y_{-i})$, which is the amount of $i$'s budget used, may decrease sharply as $y_i$ decreases from 1 to 0. This issue is somewhat mitigated in $\falg$, where $\falg_i(y_i,Y_{-i})$ increases as $y_i$ decreases but the rate at which it increases may be quite small in comparison to what the ``dominance" property requires. 
	We overcome this challenge by finding new structural insights into the problem that we use to show a non-trivial lower bound on $\falg_i(y_i,Y_{-i})$. 
	To obtain this lower bound, 
	our key new insight is to link 
	$\falg_i(y_i,Y_{-i})$ with the total increase in bid prices from the base scenario i.e., 
	\[\sum_{t\in T} (\lambda_t(y_i,Y_{-i})-\lambda_t(1,Y_{-i})).\] 
	In particular, we show that any overall gain in bid prices originates from the budget utilization of $i$ (Lemma \ref{alglbl}).  {\color{black} We also give an instance (Example \ref{tight}) where only the budget utilization of $i$ changes with $y_i$ and this new lower bound on $\falg_i(y_i,Y_{-i})$ is tight.} Importantly, this result is only true for $\falg$ and it is not true for \alg. The absence of this structural result for \alg\ is the main reason that we first analyze $\falg$ in order to show a guarantee for \alg. 
For more details, see Appendix \ref{appx:whyfrac} after reading Section \ref{sec:ana}. 
	
	In addition to the lower bound on $\falg_i(y_i,Y_{-i})$, we show that if $i$'s budget is unused for some $y_i<1$, then the bid prices have to increase by a certain amount (Lemma \ref{sum}). By combining these two properties, 
	we establish \eqref{dual2} for $\falg$ with a sufficiently large value of $\alpha$. 

	Now, we demonstrate that $\alg_i(Y)$ may decrease as $y_i$ decreases with a simple example where individual bids are comparable to the overall budget. In Appendix \ref{appx:challenge}, we give a similar example in the small bids regime (see Example \ref{integfrac}) and also an example to show that the ``dominance" property does not hold for $\falg$ (see Example \ref{fracbid} after reading Section \ref{sec:ana}).} 
	
	\smallskip
	
	\begin{eg}\label{ex}\emph{
			Consider an instance with two resources and starting budgets $B_1=2$ and $B_2=1$. We have two arrivals, $t$ and $t+1$. Arrival $t$ has bids $b_{1,t}=b_{2,t}=1$. Arrival $t+1$ has bids $b_{1,t+1}=2$ and $b_{2,t+1}=4$. Consider an execution of \alg\ with $\beta=1$ and the random seed $y_2$, for resource $2$, fixed at $0.5$. 
			Observe that, 
			\begin{enumerate}[(i)]
				\item $b_{1,t+1}(1-g(0))) < b_{2,t+1} (1-g(0.5))$. 
				\item $b_{1,t}(1-g(0.5)) = b_{2,t} (1-g(0.5))$.
			\end{enumerate}
			Consider the matching generated by $\alg$ as we decrease $y_1$ from $1$ to 0. For $y_1> 0.5$, \alg\ matches $t$ to resource $2$ and this uses up all of $2$'s budget. Consequently, $t+1$ is matched to resource $1$. For $y_1<0.5$, $t$ is matched to $1$ and $t+1$ to $2$. Thus, the amount of resource $1$ matched in $\alg$ decreases from 2 to 1 as $y_1$ decreases. This is somewhat surprising as the bid price of resource 1 increases when we decrease $y_1$. At a high level, there are two main reasons behind this phenomenon.
		}	
	\end{eg}
	%
	The first (and main) reason is that bid prices vary with both time and seed values $Y$. In the example above, for $y_1<0.5$, the bid price of resource $1$ exceeds that of resource $2$ at arrival $t$ but the time variation in bids ensures that the bid price of resource 2 dominates that of resource 1 at arrival $t+1$. 
	A second reason is that due to unknown budgets, \alg\  may overestimate the available budget of a resource. In the example above, recall that only 1 unit of resource $2$ is available at $t+1$. For $y_1<0.5$, we have $b_{1t+1}r_1(1-g(0.5)))> r_2 (1-g(y_2)),$ i.e., if we account for the remaining budget of each resource, then resource $1$ has a higher bid price at arrival $t+1$ than resource 2. However, \alg, being ignorant of the budget, 
	computes a higher bid price for resource $2$ at $t+1$. Observe that $\falg$ is not budget oblivious and never overestimates the available budget of a resource. However, even in $\falg$, for $y_1<0.5$, the time variation in bids ensures that when $t+1$ arrives, the bid price of resource 2 is higher than that of resource 1 and a fraction of arrival $t+1$ is matched to resource 2. 

\subsection{Lower Bound on Competitive Ratio of $\falg$ (Theorem \ref{guafalg})}\label{sec:ana}
{\color{black} Recall that, to lower bound the competitive ratio of $\falg$, it suffices find a feasible solution to the system \eqref{dual1}$-$\eqref{dual3} with a suitably large value of $\frac{\alpha}{1+\varepsilon}$. Also recall the candidate solution \eqref{dualdef} given by $\lambda_t = E_Y\left[\lambda_\tau(Y)\right]$ and $\theta_i=E_Y[\theta_i(Y)]$, where  
\begin{eqnarray*}
	&& \lambda_t(Y)= 
	\sum_{j\in \falg_t(Y)} b_{j,t}\,\delta_{j,t}(Y)(1-g(y_j))\qquad  \forall t\in T,\, Y\in [0,1]^n,\\
	&& \theta_i(Y) =\sum_{t\in \falg_i(Y)} b_{i,t}\,\delta_{i,t}(Y)\, g(y_i)\,=\,\falg_i(Y)\, g(y_i)\qquad \forall i\in I,\, Y\in [0,1]^n.
\end{eqnarray*}
Observe that on every sample path $Y$, $\lambda_t(Y)>0\,\, \forall t\in T$ and $\theta_i(Y)>0\,\, \forall i\in I$. Thus, constraints \eqref{dual3} are  satisfied. 
\begin{lemma}\label{dualreward}
	For the candidate solution given by \eqref{dualdef}, we have,
	\[\sum_{i\in I} \theta_i+\sum_{t\in T} \lambda_t = \falg, \]
	and	constraint \eqref{dual1} 
	is satisfied with $\varepsilon=0$. 
\end{lemma}
\begin{proof}{Proof.}
	Fix an arbitrary seed $Y$ and consider an arrival $t\in T$. The fraction of $t$ matched to resource $i$ in $\falg(Y)$ is given by $\delta_{i,t}(Y)$. This fractional match contributes a reward of $b_{i,t}\, \delta_{i,t}(Y)$ to $\falg(Y)$. By definition of $\theta_i(Y)$ and $\lambda_t(Y)$, the reward ($b_{i,t}\, \delta_{i,t}(Y)$) is split into two parts such that a fraction $g(y_i)$ of the reward goes to $\theta_i(Y)$ and the remaining $1-g(y_i)$ fraction goes to $\lambda_t(Y)$. From this observation, it follows that,
	\[\sum_{t\in T} \lambda_t(Y)+\sum_{i\in I} \theta_i(Y)=\sum_{i\in I, t\in T} b_{i,t}\, \delta_{i,t} (Y)=\falg(Y).\]
	Taking expectation over $Y$ on both sides completes the proof.
	\hfill\Halmos\end{proof}

It remains to show that 
inequalities \eqref{dual2} are satisfied. 
In fact, we show a stronger statement as described in the following lemma. 
\begin{lemma}
	Consider a resource $i\in I$ and let seed $Y_{-i}=(y_j)_{j\in I\backslash\{i\}}$ denote the random seed in $\falg$ for all resources except $i\in I$. Suppose that for the candidate solution \eqref{dualdef},  we have,
	\begin{equation}\label{dual2c}
		E_{y_i}\left[ \theta_i(Y)+\sum_{t\in \opt_i}\lambda_t(Y) \,\big|\, Y_{-i}\right]\geq \alpha\, \opt_i\quad \forall\, Y_{-i}\in[0,1]^{n-1},
	\end{equation}
	for some value $\alpha>0$. Then, inequality \eqref{dual2} is satisfied for resource $i$ with the same $\alpha$. 
\end{lemma}
\begin{proof}{Proof.} The lemma follows by taking expectation over $Y_{-i}$ on both sides of \eqref{dual2c}. 
	
	\hfill\Halmos\end{proof}

\subsubsection{Notation and Definitions.}\label{sec:addnote} We fix an arbitrary resource $i \in I$, seeds $Y_{-i}$, and establish \eqref{dual2c} for a sufficiently large value $\alpha$. For brevity, we suppress dependence on $Y_{-i}$ from notation and highlight only the dependence on seed $y_i\in[0,1]$. So $\lambda_t(y_i,Y_{-i})$ and $\theta_j(y_i,Y_{-i})$ are denoted as $\lambda_t(y_i)$ and $\theta_j(y_i)$ respectively. Further, $\falg(y_i)$ is the matching generated by $\falg$ when it is executed with seed $Y=(y_i,Y_{-i})$. Similarly, $\falg_i(y_i)$ denotes the total amount of $i$'s budget allocated in $\falg(y_i)$. 
Let $I(t,y_i)$ denote the set of resources available in $\falg(y_i)$ when $t$ arrives. Let $\falg_t(y_i)$ denote the set of resource that are fractionally matched to $t$ in $\falg(y_i)$ and let $\delta_{i,t}(y_i)$ denote the fraction of $t$ matched to $i$.  

Overloading notation, we use $T$ to also denote the total number of arrivals. For the analysis, it will be useful to view $\falg$ as a continuous process during the interval $[1,T+1)$ where arrival $t$ is matched continuously over the interval $[t,t+1)$. On arrival of $t$, we start with the set $I(t,y_i)$ of resources and let $I(\tau,y_i)$ denote the set of resources available at any moment $\tau\in[t,t+1)$. During the interval $[\tau,\tau+d\tau)$ we match an infinitesimal fraction $d\tau$ of arrival $t$ to resource,
\[j_{\tau}=\underset{j\in I(\tau,y_i)}{\argmax}\,\, b_{j,t}\, (1-g(y_j)),\]
consuming an infinitesimally small amount $b_{j_{\tau}, t}\, d\tau$ of the budget of resource $j_{\tau}$. When the budget of resource $j_{\tau}$ is used up, we remove it from the set of available resources and the process continues till time $t+1$. Observe that this is an equivalent way to describe the fractional matching in $\falg$. In the continuous viewpoint, we refine the definition of $\opt_i$ to be the set of all moments $\tau\in[t,t+1)$ such that arrival $t$ is matched to $i$ in \opt. We similarly refine the definition of $\lambda_t(\cdot)$ and let
\[\lambda_\tau(y_i)= \underset{j\in I(\tau,y_i)}{\max}\,\, b_{j,t}\, (1-g(y_j))\quad \forall \tau\in[t,t+1),\, t\in T.\]
Observe that $\lambda_t(y_i)=\int_{\tau=t}^{t+1}\lambda_\tau(y_i)\, d\tau$.
To establish \eqref{dual2c}, we examine the changes in the matching generated by $\falg$ for different values of $y_i$. In particular, we compare the matching for any given value $y_i<1$, with the matching for  
$y_i=1$. The latter scenario serves as a base scenario where the bid price of $i$ is 0 everywhere. 
To enable this comparison we define some new objects. For every $i\in I, t\in T$, and $\tau\in [t,t+1)$, we define the \emph{critical threshold} $y^c_{i,\tau}\in[0,1]$ as the value for which,
\[\textbf{Critical threshold $\mb{y^c_{i,\tau}}$:}\qquad  b_{i,t}\,(1-g(y^c_{i,\tau}))=\max_{j\in I(\tau,1)} b_{j,t}\,(1-g(y_j)).\]
Set $y^c_{i,\tau}=0$ if no such value exists and $y^c_{i,\tau}=1$ if the set $I(\tau,1)$ is empty. 
Due to the monotonicity of $g(x)=e^{\beta(x-1)}$, we have a unique value of $y^c_{i,\tau}$. Consider the set of distinct critical thresholds, 
\[V:=\left\{v\mid   y^c_{i,\tau}=v \text{ for some } \tau\in \opt_i\right\}.\] 
Let
\[b(v):=\, \sum_{t \in \opt_i}b_{i,t}\int_{\tau=t}^{t+1}\onee(y^c_{i,\tau}=v)\,d\tau\,=\, \sum_{t \in \opt_i}b_{i,t}\delta_{i,t} (y_i) \qquad \forall v\in V,\]
i.e., $b(v)$ is the cumulative bid on $i$ from all moments that are in $\opt_i$ and that have critical threshold $v$. We also define
\[B(y_i)=\sum_{v\geq y_i;\, v\in V} b(v)\quad \forall y_i\in[0,1],\]
\[S(y_i)=\left\{\tau \mid y_i\leq y^c_{i,\tau} \text{ and } \tau\in \opt_i\right\}.\]
 $B(y_i)$ is the cumulative bid on $i$ from all moments in $\opt_i$ that have critical threshold at least $y_i$ and $S(y_i)$ is the set of all moments in $\opt_i$ with critical threshold at least $y_i.$  Recall that we use $\opt_i$ to also denote the total fraction of $i$'s budget that is matched in \opt. Observe that,
$B(0)=\sum_{v\in V} b(v) \opt_i$ and 
$S(0)= \opt_i.$

\subsubsection{Establishing \eqref{dual2} for the Candidate Solution.}
At a high level, our approach is to separately lower bound $\lambda_{\tau}(y_i)$ and $\theta_i(y_i)$ and then combine the two lower bounds together. We start with $\lambda_{\tau}(y_i)$ and show that the maximum bid price at $\tau$ takes its minimum value in the base scenario where $y_i=1$. 
\begin{lemma}\label{lambda}
	Given $i\in I$ and seed $Y_{-i}$, for every  $y_i\in[0,1], t\in T,$ and $\tau\in [t,t+1)$, we have
	\[\lambda_\tau(y_i)\, \geq\, \lambda_{\tau}(1)\,\geq\,  b_{i,t}(1-g(y^c_{i,\tau})).\]
\end{lemma}
The proof is included in Appendix \ref{appx:lambda}. Next, to lower bound $\theta_i(y_i)$, we establish a lower bound on $\falg_i(y_i)$. 
Recall that, $\falg_i(y_i)$ is the set of arrivals matched to $i$ in $\falg(y_i)$, as well as, the total budget of $i$ used in $\falg_i(y_i)$.	When bids $b_{j,t}\in\{0,b_j\times b_t\}$ for every $j\in I,\, t\in T$, we show in Appendix \ref{appx:decompose} that $\falg_i(y_i)\geq B(y_i)$, a lower bound that matches the ``dominance" property. 
In general, $\falg_i(y_i)$ may be strictly smaller than $B(y_i)$ with non-zero probability (see Example \ref{fracbid}). In fact, there are instances where $\frac{\falg(y_i)}{B(y_i)}\to 0$ for some $y_i\in(0,1)$ (see Example \ref{tight}). We give a novel lower bound on $\falg_i(Y)$ by linking 
it with the increase in bid prices over all arrivals. In particular, we show that any overall gain in bid prices originates from the budget utilization of $i$. 

\begin{lemma}\label{alglbl}
	Consider a resource $i\in I$ and seed vector $(y_i,Y_{-i})\in[0,1]^n$ such that $\falg_i(y_i)<B(y_i)$. Then, we have $\falg_i(y_i)\geq \sum_{v\in V,\, v\geq y_i}b(v)\, \frac{g(v)-g(y_i)}{1-g(y_i)}$.
\end{lemma}
\begin{proof}{Proof.}
	Given resource $i$ and seed $Y_{-i}$, for every $y_i\in [0,1]$, define
	\[\lambda_{net}(y_i)= \int_{\tau=1}^{T+1}\lambda_\tau(y_i)\, d\tau\,=\,\sum_{t\in T} \lambda_{t}(y_i). \]
	Using Lemma \ref{lambda} we have, 
	$\lambda_{net}(y_i)\geq \lambda_{net}(1)\,\,\forall y_i\in[0,1]$.
	Now, fix a seed $y_i$ such that $\falg_i(y_i)<B(y_i)$. Observe that the main claim follows from the following upper and lower bounds on $\lambda_{net}(y_i)-\lambda_{net}(1)$.
	\[\falg_i(y_i)\,(1-g(y_i))\, \geq\,\lambda_{net}(y_i)-\lambda_{net}(1)\,\geq\, \sum_{v\in V,\,v\geq y_i} b(v) (g(v)-g(y_i)).\]
	\medskip
	
	\noindent \textbf{Proof of lower bound:}  Since $\falg_i(y_i)<B(y_i)\leq B_i$, we have that $i$ is available at every moment in $\falg(y_i)$ 
	and at every $\tau\in S(y_i)\cap[t,t+1)$, 
	$\falg$ selects a resource $j\in I$ such that,
	\[b_{j,t}\,(1-g(y_j)) \geq b_{i,t}\,  (1-g(y_i)).\]
	From this, we have for every  $\tau\in S(y_i)\cap[t,t+1)$,
	\[\lambda_\tau(y_i) -\lambda_\tau(1)\geq b_{i,t}\,  (g(y^c_{i,\tau})-g(y_i)),\]
	where we used the following facts (i) For $y^c_{i,\tau}>0$, we have $\lambda_\tau(1)= b_{i,t}\,  (1-g(y^c_{i,\tau}))$ and (ii) For $y^c_{i,\tau}=0$ and  $\tau\in S(y_i)\cap[t,t+1)$, we have $y_i=0$ and therefore, $g(y^c_{i,\tau})-g(y_i)=0$.
	
	Integrating over all moments in $S(y_i)$, we get 
	\[\int_{\tau\in S(y_i)}\left(\lambda_\tau(y_i)-\lambda_\tau(1)\right)\, d\tau \,\geq \, \sum_{v\in V,\, v\geq y_i}b(v)\left(g(v)-g(y_i)\right).\]
	Finally, from Lemma \ref{lambda} we have that $\lambda_{net}(y_i)-\lambda_{net}(1)\geq\int_{\tau\in S(y_i)}\left(\lambda_t(y_i)-\lambda_t(1)\right)\,d\tau$, completing the proof of the lower bound.
	\medskip
	
	\noindent \textbf{Proof of upper bound:}  We start by observing that for every seed $y\in[0,1]$ of resource $i$,
	\[\lambda_{net}(y)=\falg_i(y)\,(1-g(y))+\sum_{j\in I\backslash\{i\}} \falg_j(y)\, (1-g(y_j)). \]
	Therefore, 
	\[\lambda_{net}(y_i)-\lambda_{net}(1)=\falg_i(y_i)\,  (1-g(y_i))+\sum_{j\in I\backslash\{i\}} \left[\falg_j(y_i)-\falg_j(1)\right]\, (1-g(y_j)), \]
	where we used the fact that $g(1)=1$ for every value of $\beta>0$.
	Now, the desired upper bound on $\lambda_{net}(y_i)-\lambda_{net}(1)$ follows from the claim that, 
	\[\falg_j(y_i)\leq \falg_j(1)\qquad \forall j\in I\backslash\{i\}.\] 
	Since $\falg_j(y_i)\leq B_i$, the claim is obviously true when $\falg_j(1)=B_j$, so let $\falg_j(1)<B_j,$ i.e., resource $j$ is available at every moment in $\falg(1)$. Therefore, in $\falg(1)$, every moment $\tau\in [t,t+1)$ that is not matched to $j$, must be matched to a resource with (strictly) higher bid price (since ties between bid prices do not occur except on a probability 0 set of seed values), i.e., 
	\[\lambda_{\tau}(1)> b_{j,t}\,  (1-g(y_j)).\]
	Since $\lambda_\tau(y_i)\geq \lambda_\tau(1)$ (from Lemma \ref{lambda}), we have that every moment that is not matched to $j$ in $\falg(1)$, is not matched to $j$ in $\falg(y_i)$ either. Therefore, $\falg_j(y_i)\leq \falg_j(1)$. 
	
	\hfill\Halmos\end{proof}
In Appendix \ref{appx:whyfrac}, we discuss why Lemma \ref{alglbl} does not hold for \alg\ when we consider the natural candidate solution given by \eqref{lambda1}---\eqref{dualdef}. For $\falg$, the lower bound in Lemma \ref{alglbl} is tight and we give an example below. 
\begin{eg}[Tightness of Lemma \ref{alglbl}]\label{tight}\emph{ Consider an instance with resources $i\in[n]$ and arrivals $t\in[2n-1]$. Let 
		\[B_i=(1-\epsilon)^{i-1}\quad \forall i\in[n-1]\quad \text{ and  } \quad B_n= 
		\frac{1-(1-\epsilon)^n}{\epsilon}.\] 
		Observe that $B_n=\sum_{i\in[n-1]}B_i$. For $i\in[n-1]$, resource $i$ has non-zero bids from arrivals $i,i+1,$ and $n+i$ such that,
		\[b_{i,i}=b_{i,i+1}=b_{i,n+i}=(1-\epsilon)^{i-1}.\] 
		Resource $n$ has non-zero bids from the first $n$ arrivals, \[b_{n,t}=(1-\epsilon)^{t-1}\quad \forall t\in[n].\] 
		\opt\ matches resource $n$ to the first $n$ arrivals and matches resource $i$ to arrival $n+i,\,\, \forall i\in[n-1]$. This fully utilizes the initial budget of every resource. It is not hard to see that $\falg$ outputs an integral matching for all seeds $Y$. Consider $\falg$ with seeds $y_i=0.5\,\, \forall i\in[n-1]$. Notice that, $y^c_{n,t}=0.5\,\, \forall t\in [n]$. Now, $B(y_n)=B_n$ for $y_n<0.5$ but for small enough $\delta>0$, when $y_n=0.5-\delta$, $\falg$ matches arrival 1 to resource $n$ and matches arrival $t$ to resource $t-1\,\, \forall t\in [n]\backslash\{1\}$. 
		Thus, $\falg_n(0.5-\delta)=\frac{\epsilon}{1-(1-\epsilon)^n}B_n$. 
		For $\epsilon=\frac{\beta\delta}{e^{0.5\beta}-1}$, $\epsilon=\frac{1}{\sqrt{n}}$, and $n\to+\infty$,  we have
		\[\frac{\epsilon}{1-(1-\epsilon)^n}B_n\to 
		\epsilon B_n\quad \text{ and }\quad \frac{g(0.5)-g(0.5-\delta)}{1-g(0.5-\delta)}B_n\to \epsilon B_n,\]
		which shows that the lower bound in Lemma \ref{alglbl} is tight. 
}\end{eg}
We now combine Lemma \ref{lambda} and Lemma \ref{alglbl} to give a lower bound on $\theta_i(Y)+\sum_{t\in \opt_i}\lambda_t(Y)$. 

\begin{lemma}\label{sum}
	Given $i\in I$ and seed $Y_{-i}$, let $\onee(y_i\leq v)$ indicate the event $y_i\leq v$. Then, for every $y_i\in[0,1]$, we have 
	\begin{eqnarray}
		&&
		\theta_i(y_i)
		+\sum_{t\in \opt_i}\lambda_t(y_i)\nonumber\\
		&&\quad \geq\,  \sum_{v\in V}b(v)\Bigg[ 1-g(v)+
		\onee(y_i\leq v)\, \left( \min\left\{g(y_i),\frac{g(v)-g(y_i)}{1-g(y_i)} \right\}\right) \Bigg].\nonumber
	\end{eqnarray}
\end{lemma}
\begin{proof}{Proof.}
	Consider $i\in I$ and $Y_{-i}\in[0,1]^{n-1}$ and fix an arbitrary seed $y_i\in[0,1]$. 
	As discussed earlier, $\falg_i(y_i)$ may be much smaller than $B(y_i)$. 
	Keeping this in mind, we consider two cases. In the first case, all if $i$'s budget is used and $\falg_i(y_i)$ is sufficiently large. In the second case, some of $i$'s budget is unused and we show that the bid prices have to increase by a certain amount. 
	Combining the two cases 
	will give us the desired.
	
	\noindent \textbf{Case I}: $\falg_i(y_i)\geq B(y_i)$. Thus,
	\[	\theta_i(y_i)\geq B(y_i) g(y_i) = \sum_{v\in V} \onee(y_i\leq v)\, b(v)\,  g(y_i).\]
	Combining this with the lower bound from Lemma \ref{lambda}, we have
	\begin{equation}\nonumber 
		\sum_{t\in \opt_i} \lambda_t(y_i)+ \theta_i(y_i)\,\geq \,\sum_{t\in \opt_i}\int_{\tau=t}^{t+1} \lambda_\tau(1)\, d\tau + \theta_i(y_i)\, \geq \,
		\sum_{v\in V} b(v)\,  (1-g(v)+\onee(y_i\leq v)\,g(y_i)) . 
	\end{equation}
	\smallskip
	
	\noindent \textbf{Case II:} $\falg_i(y_i)<B(y_i)$ ($\leq B_i,$) i.e., resource $i$ is available at every moment in $\falg(y_i)$. First, we establish a refined lower bound on $\sum_{t\in \opt_i}  \lambda_{t}(y_i) $, followed by a lower bound on $\theta_i(y_i)$. 
	Since $i$ is available at every moment in $\falg(y_i)$, we have,
	\begin{equation}\nonumber
		\lambda_\tau(y_i) \geq b_{i,t}\, (1-g(y_i))\,=\,b_{i,t}\, \left(1-g(y^c_{i,\tau})+ g(y^c_{i,\tau})-g(y_i)\right)\quad \forall t\in T,\, \tau\in S(y_i)\cap[t,t+1). 
	\end{equation}
	Using this we have,
	\begin{eqnarray}
		\sum_{t\in \opt_i} \lambda_t(y_i)\geq 
		\sum_{v\in V} b(v) \left[1-g(v) + \onee(y_i\leq v) \left(g(v)-g(y_i) \right) \right].\label{inclamb}
	\end{eqnarray}
			Now, 
			$\theta_i(y_i)=\falg_i(y_i)\, g(y_i)$.
			From Lemma \ref{alglbl}, 
			\begin{eqnarray}
				\falg_i(y_i)\, g(y_i)
				&\geq & \sum_{v\in V;\, v\geq y_i} b(v)\, \frac{g(v)-g(y_i)}{1-g(y_i)}\, g(y_i).\label{balance}
			\end{eqnarray}
			
			\noindent Combining \eqref{inclamb} and \eqref{balance} we get,
			\begin{eqnarray*}
				&&	
				\theta_i(y_i)+\sum_{t\in \opt_i} \lambda_t(y_i)  \\
				&&\qquad \geq 	 \sum_{v\in V} b(v) \left[1-g(v) + \onee(y_i\leq v) \left(g(v)-g(y_i) +\frac{g(v)-g(y_i)}{1-g(y_i)}\, g(y_i)\right) \right]\\
				&&\qquad \geq 	\sum_{v\in V} b(v) \Big[1-g(v) + \onee(y_i\leq v) \frac{g(v)-g(y_i)}{1-g(y_i)} \Big]
			\end{eqnarray*}
			
			For any given $y_i$, we could be in the worse of the two cases above, resulting in the following combined lower bound,
			\begin{eqnarray*}
				\theta_i(y_i)+\sum_{t\in \opt_i} \lambda_t(y_i) 
				\geq 	\sum_{v\in V} b(v) \Big[1-g(v) +\onee(y_i\leq v)\, \min\left\{ g(y_i), \frac{g(v)-g(y_i)}{1-g(y_i)}\right\} \Big].
			\end{eqnarray*}
			\hfill\Halmos\end{proof}

		\begin{proof}{Proof of Theorem \ref{guafalg}.}
			Let $\alpha= \min_{v\in[0,1]} \left[1-g(v)+
			\int_0^v \left( \min\left\{g(y),\frac{g(v)-g(y)}{1-g(y)} \right\}\right)dy\right]$. Taking expectation over the randomness in seed $y_i$ on both sides of the inequality in Lemma \ref{sum}, we have
			\begin{eqnarray*}
				&&E_{y_i}\left[\theta_i(y_i)+\sum_{t\in \opt_i} \lambda_t(y_i) \mid Y_{-i}\right]\nonumber\\
				&&\quad \geq\,  \sum_{v\in V}b(v)\Bigg[ 1-g(v)+
				\int_0^v \left( \min\left\{g(y),\frac{g(v)-g(y)}{1-g(y)} \right\}\right) dy \Bigg],\nonumber\\
				&&\quad \geq \alpha\, \sum_{v\in V}b(v)\, =\, \alpha\, r_i \, \opt_i.
			\end{eqnarray*}
			It remains to lower bound $\alpha$. We do this step numerically. For $g(x)=e^{x-1}$, we obtain $\alpha>0.508$ (minimum at $x=0.586$) and for $g(x)=e^{1.15(x-1)}$ we obtain $\alpha>0.522$ (minimum at $x=0.789$, see Figure 1). %
			
			\hfill\Halmos\end{proof}
		\noindent \textbf{Remark:} As stated in Theorem \ref{decomp}, \alg\ with $\beta=1$ is $(1-1/e)$ competitive for special cases of Adwords. 
		We present the proof in Appendix \ref{appx:decompose}.
		\begin{figure}[h]\label{plott}
			\includegraphics[scale=0.4]{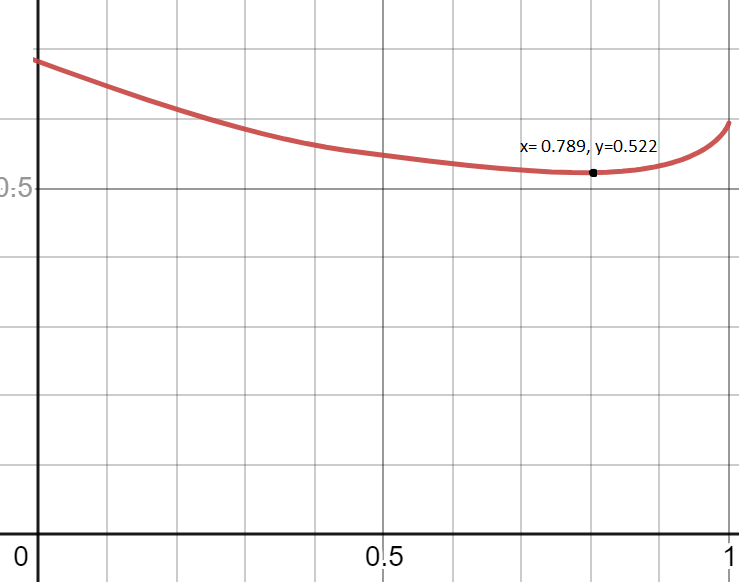}
			\centering
			\caption{Plot of $\alpha(x)= \left[1-g(x)+
				\int_0^x \left( \min\left\{g(y),\frac{g(x)-g(y)}{1-g(y)} \right\}\right)dy\right]$ for $g(x)=e^{1.15(x-1)}$.}
		\end{figure}
	
		\subsection{From Fractional to Integral Algorithm (Theorem \ref{main})}\label{sec:fracint}
		
		Since \alg\ is oblivious to budgets but $\falg$ is not, 
		there can be a substantial difference 
		between the output of these algorithm for the same seed $Y$ (see Example \ref{integfrac} in Appendix \ref{appx:challenge}).  Therefore, we compare \alg\ not with $\falg$ on the same instance but with the performance of $\falg$ on a modified instance where the budgets of resources are increased as described next.
		\smallskip
		
		\noindent \emph{Budget augmentation:} Consider an instance where the budget of every item is augmented as follows, 
		\[B^a_i=B_i+\max_{t\in T} b_{i,t} \quad \forall i\in I.\]
	For instances with bid-to-budget ratio $\gamma$, observe that,
	\[\frac{B^a_i}{B_i}\leq 1+\gamma \quad \forall i\in I.\]
		
		\begin{proof} {Proof of Theorem \ref{main}.}
			Let $\opt^a$ denote the total reward of optimal offline on an instance with augmented budgets $B^a_i$. Let $\falg^a$ denote the expected total reward of $\falg$ on the augmented instance.	For any seed $Y$, let $\alg(Y)$ denote the total reward of \alg\ as well as the matching output by \alg\ on the original instance with seed $Y$. Similarly, let	 $\falg^{a}(Y)$ denote the total reward of $\falg$ as well as the matching output by $\falg$ on the augmented instance with seed $Y$. For any resource $i\in I$, let $\falg^a_i(Y)$ denote the total revenue from resource $i$ in $\falg$ on the augmented instance with seed $Y$. Similarly, let $\alg_i(Y)$ denote the total revenue from $i$ in $\alg$ on the original instance with seed $Y$. Finally, let $x^{a,f}_j(t,Y)$ denote the remaining budget of resource $j\in I$ when $t$ arrives in $\falg^a(Y)$. Similarly, let $x_j(t,Y)$ denote the remaining budget of $j$ after arrival $t-1$ in $\alg(Y)$. 
			
			Now, fix an arbitrary seed $Y$ and a resource $i\in I$. We claim that,
			\begin{equation}\label{itemwise}
				\frac{B^a_i}{B_i}\, \alg_{i}(Y)\geq \falg^{a}_{i}(Y).
			\end{equation}
				Before proving \eqref{itemwise}, we show that \eqref{itemwise} suffices to prove the theorem. Taking expectation over seeds $Y$ and also a summation over $i\in I$ on both sides of \eqref{itemwise} gives us,
					\begin{equation}\label{pathwise}
					\alg\geq \frac{1}{1+\gamma}\,\, \falg^{a},
				\end{equation}
				here we use the fact that $\frac{B^a_i}{B_i}\leq 1+\gamma\,\, \forall i \in I$. Since an allocation that is feasible in the original instance is also feasible after augmenting the budgets, we have
			\begin{equation*}
				\opt^a\geq \opt.
			\end{equation*}
			Given a lower bound of $\eta$ on the competitive ratio guarantee of $\falg$ we have,
			\begin{equation}\label{fa}
				\falg^a\, \geq\, \eta \, \opt^a\,\geq\, \eta\, \opt.
			\end{equation}
		Combining \eqref{pathwise} and \eqref{fa} proves the theorem.
			It remains to show \eqref{itemwise} and we establish it in cases.
			\smallskip
			
			\noindent \textbf{Case I:} All $B_i$ of $i$'s budget is allocated in $\alg(Y)$ i.e., $\alg_i(Y)=B_i$. Then, \eqref{itemwise} follows by observing that $\falg^{a}_{i}(Y)\leq B^a_i\,=\, B_i^a\, \frac{\alg_i(Y)}{B_i}$. 
			\smallskip
			
			\noindent \textbf{Case II:} Some of $i$'s budget is not allocated to any arrival in $\alg(Y)$. In this case, we show more strongly that $\falg^a_{i}(Y)\leq \alg_i(Y)$. Define 
			\begin{eqnarray*}
				z^{a,f}_{j}(t,Y)&=&\max\{0,B_j-x^{a,f}_j(t,Y)\}=(B_j-x^{a,f}_{j}(t,Y))^+\quad \forall j\in I,\\
				z_j(t,Y)&=&(B_j-x_j(t,Y))^+\quad \forall j\in I.
			\end{eqnarray*}
			Notice that we use the original budget $B_i$ for defining $z^{a,f}_{j}(\cdot,\cdot)$ as well as $z_j(\cdot,\cdot)$.  
			As some of $i$'s budget is unused in $\alg(Y)$ after the final arrival $T$, we have $z_i(T+1,Y)>0$. 
			
			To prove that $\falg^a_{i}(Y)\leq \alg_i(Y)$, it suffices to show that $z^{a,f}_{i}(T+1,Y)\geq z_{i}(T+1,Y)$.
			In fact, we show more strongly that 
			\begin{equation}\label{morestrong}
				z^{a,f}_{j}(t,Y)\geq z_{j}(t,Y)\quad \forall j\in I, t\in T.
			\end{equation}
			This is true at $t=1$. Suppose this is true just prior to arrival of $t$. We will show that the inequality also holds just prior to arrival of $t+1$. Then by induction we have the desired.
			
			Given $z^{a,f}_{j}(t,Y)\geq z_{j}(t,Y)$, we have 
			\begin{equation}\label{dominance}
				\{j\mid z_j(t,Y)>0\}\,\subseteq\, \{j\mid z^{a,f}_j(t,Y)>0\},
			\end{equation}
			where the LHS is the set of resources available in $\alg(Y)$ at $t$. Let $e$ be the resource matched to $t$ in $\alg(Y)$. 
			 Recall that when $\falg$ matches an arrival to multiple resources, it matches fractions of the arrival to each resource in descending order of bid prices. Consider resources in the set $\{j\mid z_j(t,Y)>0\}$ and let $w$ denote the first resource (if any) in this set that is fractionally matched to $t$ in $\falg^a(Y)$. 
			We consider two cases based on whether $w$ exists.
			 
			 \noindent \textbf{Case II.A:} Resource $w$ does not exist, i.e., no resource in the set $\{j\mid z_j(t,Y)>0\}$ is matched to $t$ in $\falg^a(Y)$. Then, 
			 \begin{eqnarray*}	 
			 	z^{a,f}_q(t+1,Y)&=&z^{a,f}_q(t,Y)\geq z_q(t,Y) \geq z_q(t+1,Y)\quad \forall q\in \{j\mid z_j(t,Y)>0\},\\
			 	z^{a,f}_q(t+1,Y)&\geq&0=z_q(t,Y)=z_q(t+1,Y)\quad \forall q\in \{j\mid z_j(t,Y)=0\}.
			 \end{eqnarray*}
				\noindent \textbf{Case II.B:} Resource $w$ exists. Since both $\alg$ and $\falg$ match greedily according to bid prices, we have
					\[b_{w, t}\,(1-g(y_{w}))=\max_{\{j\mid j\in z_j(t,Y)>0\}}b_{j,t} (1-g(y_j))=\, 
					b_{e,t}\, (1-g(y_e)),\]
			 here the first equality follows from \eqref{dominance} and the fact that $w$ is the first resource in $\{j\mid z_j(t,Y)>0\}$ that is matched to $t$ in $\falg^a(Y)$. The second equality follows by definition of $e$. Since the bid prices of two different resources are unequal w.p.\ 1, we have that,
			 \[w=e.\]
			 Now, using $z^{a,f}_e(t,Y)\geq z_e(t,Y)$ and the fact that $t$ is fractionally matched in $\falg(Y)$, we have,
			  \[z^{a,f}_{e}(t+1,Y)\geq z_e(t+1,Y).\] 
			Further, since $z^{a,f}_{e}(t,Y)> 0$, at least $\max_{t\in T} b_{e,t}$ of $e$'s augmented budget is available in $\falg^a(Y)$ when $t$ arrives. Hence, in $\falg^a(Y)$, no resource from the set $\{j\mid j\neq e, z_j(t,Y)>0\}$ is matched (fractionally) to $t$, i.e., \[z^{a,f}_q(t+1,Y)=z^{a,f}_q(t,Y)\geq z_q(t,Y)=z_q(t+1,Y)\qquad \forall q\in \{j\mid j\neq e, z_j(t,Y)>0\}.\] 
			Similar to Case II.A, we also have,
\[				z^{a,f}_q(t+1,Y)\geq0=z_q(t,Y)=z_q(t+1,Y)\quad \forall q\in \{j\mid z_j(t,Y)=0\}.\]
			This completes the proof.
\hfill\Halmos		\end{proof}
		\noindent \textbf{Remark:} When every bid is either 0 or 1 ($b$-matching case), it is easy to see that $\falg$ and \alg\ are identical on every sample path. Since $b$-matching is a special case of Adwords with decomposable bids ($b_{i,t}\in\{0,b_i\times b_t\}\,\, \forall i\in I, t\in T$), and $\falg$ is $(1-1/e)$ competitive for Adwords with decomposable bids (Appendix \ref{appx:decompose}), we have that, \alg\ is $(1-1/e)$ competitive for $b$-matching (with arbitrary budgets).}

\subsection{Upper Bound on Competitive Ratio of \alg}\label{sec:randub}
The upper bound of $(1-1/e)$ for online bipartite matching and its many generalizations (including Adwords), follows from a family of instances where arrivals get pickier over time; the first arrival has an edge to every resource and for every subsequent arrival, one edge is dropped uniformly randomly, i.e., the $t$-th arrival has exactly $n-t+1$ edges and the $t+1$-th arrival has an edge to all but one (randomly chosen) neighbor of arrival $t$~\citep{kvv}. This family of instances has identical bids and budgets; features that are also sufficient for a lower bound of $(1-1/e)$ on the competitive ratio of \alg\ (Theorem \ref{decomp}). 
To obtain a tighter upper bound on GPG, we start by constructing an instance where bids are not decomposable and GPG is at most $(1-1/e)\, \opt$. 
Then, by a suitable modification of this instance, we obtain a family of instances such that for every $\beta\geq 1$, there exists an instance in the family where \alg\ is strictly less than $(1-1/e)\, \opt$. 


To describe the main ideas we let $\beta=1$ and allow arbitrary bid to budget ratio. Consider an instance with $n+1$ resources and $2n$ arrivals. Each arrival bids $1$ on every resource in $\{1,2,\cdots,n\}$. Arrival $t\in[n]$ bids \[w(t)=\frac{1-e^{\frac{t}{n+1}-1}}{1-e^{-1}},\] 
on resource $n+1$. Arrivals $t\geq n+1$ do not have an edge to resource $n+1$. Resources $\{1,2,\cdots,n\}$ have unit budget and resource $n+1$ has a budget of $\sum_{t\in[n]} w(t)$. The optimal matching allocates resource $n+1$ to every $t\in[n]$ and matches arrival $t$ to resource $t-n$ for $t\geq n+1$. We have,
\[\opt=n+\sum_{t\in[n]} w(t),\]
here $\lim_{n\to+\infty}\frac{1}{n}\sum_{t\in[n]} w(t)\to \frac{1}{e-1}$. Therefore, $\lim_{n\to+\infty}\frac{1}{n}\opt\to \frac{e}{e-1}$. 

Let $y_i$ denote the random seed of resource $i\in[n+1]$ in \alg. 
For $k\in[n]$, let $z_k$ denote the $k$-th order statistic of the random variables $\{y_i\}_{i\in[n]},$ i.e., the $k$-th smallest seed value. For $k\in[n]$, random variable $z_k$ has mean $\frac{k}{n+1}$. Let $I(t)$ denote the (random) set of resources available at $t$ in \alg. 
Resource $n+1$ is matched to arrival $t\in[m]$ if \[w(t)(1-e^{y_{n+1}-1})>\max_{i\in I(t)} (1-e^{y_i-1})\geq (1-g(z_t)),\]
here the last inequality follows from the fact that the $t$-th lowest seeded resource in $[n]$ will not be matched to any of the arrivals preceding $t$. For the sake of intuition, suppose that for $n\to+\infty$, we have $z_{t}\approx E[z_{t}]=\frac{t}{n+1}\,\, \forall t\in [n]$, with high probability (w.h.p.). Then, at a high level, we argue as follows. Resource $n+1$ is matched to $t$ only if $y_{n+1}\approx 0$, which occurs with a vanishingly small probability. By linearity of expectation, the expected total consumed capacity of resource $n+1$ is $\approx 0$. Therefore, the expected value of $\frac{1}{n}\alg$ is $\approx 1$, 
and we have $\frac{\alg}{\opt}\approx(1-1/e)$.

Now, consider a family of instances parameterized by $\epsilon\geq 0$. Every instance in the family is similar to the instance above 
except that the bid from arrival $t\in[n]$ to resource $n+1$ is, 
\[w(t,\epsilon)=\frac{1-e^{\frac{t}{n+1}-1}}{1-e^{-1+\epsilon}},\] 
and the capacity of resource $n+1$ is $\sum_{t\in[n]} w(t,\epsilon)$. All other bids and parameters are left unchanged. Observe that for a fixed $t$, $w(t,\epsilon)$ is a strictly increasing function of $\epsilon$. For non-zero $\epsilon$, the value of offline (\opt) increases by 
\begin{equation*}
\sum_{t\in[n]} (w(t,\epsilon)-w(t,0))=\sum_{t\in[n]} w(t,0)\frac{e^{-1}(e^{\epsilon}-1)}{1-e^{-1+\epsilon}}
=O(\epsilon n).
\end{equation*}
To obtain an upper bound that is strictly less than $(1-1/e)$, we show that the expected total value of \alg\ increases by an amount that is significantly lower than $O(\epsilon n)$. 
 The proof is presented in Appendix \ref{appx:gpgub}. We would like to note that for \alg\ with $\beta=1$, performing Monte Carlo simulations with $n=5000$ indicates a stronger upper bound of 0.604 when bids from arrivals $t\in[n]$ to resource $n+1$ are set as follows,
\[\hat{w}(t,\epsilon)=\frac{1-e^{\frac{t}{n+1}-1-\epsilon}}{1-e^{-1}},\]
with $\epsilon=0.2$. 

\section{Extensions and Applications}\label{sec:extension}
In this section, we demonstrate the flexibility and usefulness of budget oblivious algorithms through different settings. Our first two results are in settings that are different from Adwords. 
Then, we discuss potential applications of budget oblivious allocation in 
automated budget management for Adwords. 
\subsection{
Online Matching with Multi-channel Traffic}

Inspired by applications in two-sided matching platforms such as VolunteerMatch -- where volunteers are matched to nonprofits -- \cite{multi} introduce a setting where volunteer traffic comes from two types of sources. \emph{Internal} traffic consists of volunteers who arrive on the platform with an interest in (possibly) many nonprofits and the platform selects a set of opportunities to recommend. \emph{External} traffic models volunteers who directly arrive at a nonprofit's volunteering page via an external link that bypasses the platform's matching/recommendation system.  \cite{multi} model this as an instance of online matching where nonprofit opportunities with a finite need for volunteers correspond to offline resources with finite capacity and each volunteer corresponds to an online arrival. Up on arrival, a volunteer is probabilistically matched to at most one opportunity. External arrivals have an edge to exactly one resource. They show that when the number of external arrivals is a small fraction of the total capacity of all resources, a natural generalization of the algorithm of \cite{msvv}, which is oblivious to the source of arrivals, achieves the optimal competitive ratio for the problem. However, when external arrivals are a significant fraction of the total capacity of all resources, one can achieve strictly better performance with an algorithm that reacts to external and internal arrivals in different ways. For a deterministic version of their (stochastic) problem, we show that the budget (and arrival source) oblivious GPG algorithm achieves the best possible guarantee for every possible composition of the arrival sequence. This closes the (non-zero) gap between the best known guarantee and the best known upper bound shown for the deterministic version of the setting in \cite{multi}. 

Formally, we consider a setting with $I$ resources that have arbitrary capacities $(B_i)_{i\in I}$. Every arrival $t\in T$ comes with a set of edges to resources $(i,t)\in E$. Matching an arrival to a resource uses one unit of the resource's capacity and an arrival can only be matched to a neighboring resource with available capacity. The goal is to match the maximum number of arrivals. Every arrival is either internal or external. Let $T^\ntr$ and $T^\xtr$ denote the set of internal arrivals and external arrivals, respectively. Let $T^\xtr_i$ denote the set of external arrivals with an edge to resource $i\in I$. \cite{multi} define the \emph{fraction of external traffic} (of a given instance) as the fraction of total capacity that can be matched to external arrivals.   
\[\text{Effective Fraction of External Traffic}:
\qquad \delta=\frac{\sum_{i\in I} \min\{|T^\xtr_i|,\, B_i\}}{\sum_{i\in I} B_i}.\]

For any given $\delta$, this setting corresponds to 
OBA with binary bids and without the large capacity assumption. Therefore, 
GPG with $\beta=1$ is $(1-1/e)$ competitive for this setting (Theorem \ref{decomp}). We analyze the competitive ratio of \alg\ with $\beta=1$ as a function of $\delta$. Since $\beta=1$ leads to the best possible guarantee, we hereby fix $\beta$ at 1 and refer to \alg\ with $\beta=1$ as simply \alg. Note that, \cite{multi} consider a model where matches can fail with some probability and the starting capacities are large. We do not make the large capacity assumption but consider a deterministic setting where every match succeeds with probability one.   


Our main goal is to find how the performance of \alg\ changes with $\delta$. While \alg\ is at least $(1-1/e)$ competitive for every $\delta\in[0,1]$, at a high level, as $\delta\to 1$, the competitive ratio of \alg\ approaches 1. To see this, consider an arbitrary resource $i\in I$ and sample path $Y$ in \alg. On this sample path, $\alg$ matches every arrival in $T^\xtr_i$ provided there is some budget remaining. Thus, $\alg$ uses at least $\min\{|T^\xtr_i|, B_i\}$ of resource $i$'s budget on every sample path. 
Now, by definition of $\delta$, the expected value of \alg\ is at least $ \delta \sum_{i\in I}{B_i}\geq \delta\, \opt$. 
Overall,  we have
\[\alg\geq \max\{\delta, (1-1/e)\}\,\opt.\]
The gives a lower bound on the competitive ratio that improves as $\delta$ increases but the dependence on $\delta$ is not tight. For example, we show that \alg\ is $1+(1-1/e)\ln (1-1/e)$ ($>0.71$) competitive for $\delta=(1-1/e)$. From the upper bound shown in \cite{multi}, this is the best guarantee achievable by any online algorithm. 
\begin{repeattheorem}[Theorem 1 in \cite{multi}.] For any value of $\delta\in[0,1]$ and any resource capacities, no (randomized) online algorithm can achieve a competitive ratio better than $\max\{(1-1/e),1+\delta\ln \delta\}$.
\end{repeattheorem}
We show that the competitive ratio of \alg\ matches this upper bound for all values of $\delta\in[0,1]$. This improves the competitive ratio achieved by the novel Adaptive Capacity (AC) algorithm proposed in \cite{multi}\footnote{Unfortunately, there is no closed form expression for the guarantee of AC. See Theorem 2 in \cite{multi} for the factor revealing LP that gives a lower bound on the competitive ratio of AC. Figure 2 in \cite{multi} illustrates that there is a non-zero gap between the known competitive ratio of AC and the upper bound.}. It is perhaps surprising that, unlike the AC algorithm, \alg\ achieves the optimal guarantee without requiring any information about the channel that each arrival belongs to. {\color{black} We would like to note that this result does not require the new structural properties (such as Lemma \ref{alglbl}) that we show for the Adwords problem.}

\subsubsection{Overview of Analysis.}
In the same vein as \cite{multi}, we lower bound the competitive ratio by 
combining the LP free analysis technique with a factor revealing program (FRP). At a high level, every possible instance of the problem can be projected to a feasible solution for FRP with objective value equal to the ratio $\frac{\alg}{\opt}$ on that instance. Computing the optimal value of the FRP gives a lower bound on the worst case ratio, i.e., the competitive ratio of \alg. In general, such a program may be intractable to analyze but given the right set of constraints one can find a strong lower bound on the optimal value of the program. We derive the constraints of our FRP using the LP free framework and structural properties of \alg\ and \opt. 
We start with 
some useful structural properties of worst case instances that simplify the analysis. 

\begin{lemma} \label{simplify}
For any given $\delta\in[0,1]$,	to lower bound the competitive ratio of \alg\ it suffices to consider instances where 
\begin{enumerate}[(i)]
\item \opt\ matches every arrival. 
\item Internal traffic arrives before 
external traffic, i.e., there is no internal traffic after the first external arrival.
\end{enumerate}
\end{lemma}
We present the proof in Appendix \ref{appx:multi}. To state the FRP, let $G(x)=\int_0^x g(u) du= e^{x-1}-e^{-1}$ and notice that $G(1)=(1-1/e)$. The FRP has decision variables $\{o^\ntr_i,\, o^\xtr_i,\, x^\ntr_i,\, x^\xtr_i\}_{\forall i\in I}$ and $\zeta$. 
\begin{eqnarray}
\textbf{FRP}: \quad
&\underset{{\zeta, \{o^\ntr_i,\, o^\xtr_i,\, x^\ntr_i,\, x^\xtr_i\}_{ i\in I}}}{\inf}\quad &\zeta \nonumber\\
&\text{s.t.} & 
o^\ntr_i+o^\xtr_i\leq B_i\quad \forall i\in I, \label{cap}\\
&& x^\xtr_i\leq o^\xtr_i \quad \forall i\in I,\label{capx}\\
&& \sum_{i\in I} o^\xtr_i = \delta \sum_{i\in I} B_i,\label{fet}\\
&& \sum_{i\in I} x^\xtr_i \geq\, (\zeta-1+\delta)\,\sum_{i\in I} B_i,\label{lx}\\
& \zeta \sum_{i\in I} (o^\ntr_i+o^\xtr_i)&\geq \sum_{i\in I}\left(G(1)\, o^\ntr_i + x^\xtr_i+ o^\xtr_i    G\left(1-\frac{\underset{j\in I}{\sum}x^\xtr_i}{\underset{j\in I}{\sum}o^\xtr_i}\right)\right). 
\label{lpfree}
\end{eqnarray}
\begin{theorem}\label{crlb}\label{frpval}
For any given $\delta\in(0,1]$, the optimal value of FRP is a lower bound on the competitive ratio of $\falg$. For $\delta\geq \frac{1}{e}$, the optimal value of FRP is at least $1+\delta \ln \delta$.
\end{theorem}

%
We give an overview of the first part of the Theorem \ref{crlb} here and present the full proof of the theorem in Appendix \ref{appx:multibase}. The first part states that the optimal value of FRP is a lower bound on the competitive ratio. 
The main idea is as follows. Given an instance with fraction of external traffic $\delta$, we show that there exists a feasible solution to FRP with objective value equal to the performance ratio $\frac{\alg}{\opt}$ on that instance. Consider the following candidate solution for the FRP. 

\noindent \textbf{Candidate solution for FRP:} For resource $i\in I$, let $o^\xtr_{i}$ be the total budget of $i$ allocated to external arrivals in \opt. Let $x^\xtr_i$ be the expected budget of $i$ allocated to external arrivals in $\falg$. Let $o^\ntr_i$ and $x^\ntr_i$ be the (expected) budget of $i$ allocated to internal arrivals in \opt\ and in $\falg$, respectively. Finally, let 
\[\zeta=\frac{\sum_{i\in I} (x^\ntr_i+x^\xtr_i)}{\sum_{i\in I}(o^\ntr_i+o^\xtr_i)}.\] 
Observe that $\zeta$ is equal to the performance ratio $\frac{\alg}{\opt}$. We start by noting that for $\delta>0$, all constraints (including \eqref{lpfree}) are well defined.  In particular, for $\delta>0$, using Lemma \ref{simplify}$(i)$ we have $\sum_{i\in I} o^\xtr_i >0$ on every non-trivial instance of the problem. By using the structure of the instance and borrowing ideas from \cite{multi}, it can be shown that the candidate solution satisfies constraints \eqref{cap}--\eqref{lx}. 
The crux of our analysis is to show that inequality \eqref{lpfree} is satisfied and we give a detailed proof in Appendix \ref{sec:multianalysis}. 

\subsection{Online Matching with Stochastic Rewards}\label{sec:stochrew} 

Introduced by \cite{deb}, the problem of online matching with stochastic rewards generalizes online bipartite matching by associating a probability of success $p_{i,t}$ with every pair $i\in I, t\in T$. 
When a match is made, i.e., edge is chosen, it succeeds independently with this probability. 
If the match succeeds, the resource cannot be matched to any other arrival. If the match fails, the arrival departs but the resource is available for future rematch. The edge probabilities are revealed sequentially with each arrival. When edge probabilities are binary, i.e., $p_{i,t}\in\{0,1\}\,\, \forall i\in I, t\in T$, this setting reduces to online bipartite matching. The natural greedy algorithm that matches each arrival to an available resource with the highest success probability is 0.5 competitive for this problem~\citep{negin}. In general, this is the best known result but 
better algorithms are known for several well studied special cases~\citep{deb,mehta,stochrew,huang}. We discuss the state-of-the-art below and refer to \cite{stochrew} and \cite{huang} for a more detailed review.
\begin{enumerate}[(i)]
\item \textbf{Decomposable probabilties:}  \cite{stochrew} introduced the special case where edge probabilities are decomposable i.e., $p_{i,t}\in\{0, p_i\times p_{t}\}\,\, \forall i\in I, t\in T$. They showed that Algorithm \ref{srank} with $\beta=1$ is 
is $(1-1/e)$ competitive. \cite{when} recently gave a simpler and more general proof of this result. Note that Algorithm \ref{srank} corresponds to the GPG algorithm with bids replaced by edge probabilities. Since online bipartite matching is a special case of the setting of decomposable probabilities, $(1-1/e)$ is also the best possible competitive ratio guarantee.
\item \textbf{Vanishing probabilities:} Another well studied special case is when the edge probabilities are vanishingly small, i.e., $p_{i,t}\to 0\,\, \forall i\in I, t\in T$. \cite{stochrew} showed that a \emph{deterministic} algorithm, originally proposed by \cite{mehta} and distinct from Perturbed Greedy, has competitive ratio guarantee of at least 0.596 in this setting. 
\end{enumerate}
\smallskip

\begin{algorithm}[H]
\textbf{Inputs:} Set of resources $I$, parameter $\beta$\; 
Let $g(x)=e^{\beta(x-1)}$\;
For every $i\in I$ generate i.i.d.\ r.v.\ $y_i\sim U[0,1]$\;
\For{\text{every new arrival } $t$}{
Match $t$ to $i^*=\underset{ i\in I}{\arg\max}\,\, p_{i,t} (1-g(y_i))$\;
\textbf{if} {match succeeds} \textbf{update} $I=I\backslash\{i^*\}$\;			}

\caption{Generalized Perturbed-Greedy for Stochastic Rewards}
\label{srank}
\end{algorithm}
\smallskip

Inspired by these positive results, a perhaps natural question is if the algorithms developed for special cases point to a candidate algorithm that may outperform greedy for the general problem setting. Since the setting of stochastic rewards includes online bipartite matching as a special case, no deterministic algorithm can beat greedy in general~\citep{kvv}. This rules out all the algorithms that are known to beat greedy in the special case of vanishing probabilities. That leaves Algorithm \ref{srank}, which gives the optimal competitive ratio guarantee for decomposable probabilities (for $\beta=1$). 
Previously, it was not known if this algorithm outperforms greedy for vanishing probabilities. \cite{stochrew} highlight some challenges with proving such a result 
and leave the analysis of Algorithm \ref{srank} for vanishing probabilities as an open problem. We show that Algorithm \ref{srank} with $\beta=1$ is at least 0.508 for vanishing probabilities. 
While this does not beat the best known competitive ratio result for vanishing probabilities, it gives a single algorithm (namely Algorithm \ref{srank}) that outperforms greedy in \emph{all} previously studied special case of the problem. Recently, \cite{when} showed that in general, the competitive ratio of Algorithm \ref{srank} is strictly less than $(1-1/e)$. 
\begin{theorem}\label{stoch}
For online matching with stochastic rewards and vanishing probabilities, i.e., $\max_{i\in I, t\in T} p_{i,t}\to 0$, Algorithm \ref{srank} is 0.522 competitive for $\beta=1.15$ and $0.508$ competitive for $\beta=1$.
\end{theorem}
We present the proof in Appendix \ref{appx:stoch}. The result follows directly from an equivalence between the setting of vanishing probabilities and Adwords with unknown (stochastic) budgets. This equivalence was first observed by \cite{deb} and later generalized by \cite{huang} anc \cite{stochrew}. 
The result states an instance of online matching with vanishing probabilities is equivalent to an instance of Adwords with unknown budgets (on the same graph) with bids equal to the edge probabilities, i.e., $b_{i,t}=p_{i,t}\,\, \forall i\in I, t\in T,$ and unknown budgets 
sampled independently for each resource from the exponential probability distribution with unit mean, i.e., $B_i\sim \text{Exp}(1)\,\, \forall i\in I$. 

\subsection{Advantage of Budget Obliviousness in Automated Budget Management} \label{sec:automated}
The high level goal for a typical advertiser is to target their customers through multiple ad campaigns and marketing channels. Starting with an overall budget, an advertiser must determine a good distribution of their budget to individual ad campaigns. For search ads, advertisers must also determine the bids for relevant key words. While these decisions play a crucial role in the success of their ad campaigns, determining a good distribution of budgets between diverse options and deciding optimal bids for specific campaigns can be an incredibly challenging task for any advertiser. As a result, many advertisers rely on automated bidding and budget management tools. 
The motivation behind these tools is to improve performance for advertisers while simplifying the usage of ad platforms \citep{autobid}. 
%
Budget oblivious allocation algorithms, such as GPG, do not require fixed initial budget as input. This increases the degrees of freedom available to automated budget management tools and may allow significantly better outcomes for both the advertisers and the platform. We illustrate via a stylized example in Appendix \ref{appx:budmage}.  

\section{Conclusion}\label{sec:conclusion}
We considered the classic Adwords setting with unknown budgets and showed that a natural generalization of Perturbed Greedy algorithm, that computes random bid prices for resources without using budget information, is 0.522 competitive. This is the first result that improves on the guarantee of 0.5 obtained by the greedy algorithm. To show the result, first, we analyze the fractional version of the algorithm using recent innovations in analysis of online matching algorithms, alongside various novel structural insights. We also showed that no deterministic algorithm can do better than greedy and also gave an upper bound of 0.624 on our randomized algorithm (and observed a stronger upper bound 0f 0.604 in simulations). Finally, we demonstrated the usefulness of budget oblivious algorithms for online resource allocation in a variety of applications. An immediate open question from our work is to find out the best possible competitive ratio guarantee for the Adwords problem with unknown budgets. 

\ACKNOWLEDGMENT{We thank the anonymous referees for their valuable and helpful feedback. We are also grateful to Vijay Vazirani for many insightful discussions on the topic of this paper. Finally, we thank Will Ma for pointing us to a simpler proof of Theorem 1 and thank Sebastian Schubert for pointing out an error in an earlier version of this paper.}
\small
\bibliographystyle{informs2014.bst}
\bibliography{bib_2}
\newpage

\begin{APPENDICES}
\def\alg{\textsc{GPG}}
\section{Further Discussion of Concurrent Work}\label{appx:related}
{\color{black} \cite{albers} and \cite{vazirani, Vazirani2022TowardsAP} concurrently show that the Perturbed Greedy algorithm is $(1-1/e)$ competitive for Adwords with binary bids and arbitrary resource budgets, similar to Theorem \ref{decomp}. The overall proof structure in these results is essentially the same; each proof is based on finding a feasible solution to a linear system. Further, the candidate solutions proposed in the papers are also essentially the same and the proof of feasibility relies on establishing ``monotonicity" and ``dominance" properties similar to \cite{devanur}. The minor differences arise from the linear system used in each paper. \cite{albers} perform a primal-dual analysis based on a configuration linear program (LP) introduced for a related problem by \cite{huang}.  \cite{Vazirani2022TowardsAP} generalize the `economic viewpoint' of \cite{econ} to obtain a linear system that closely resembles that of \cite{albers} (which is stated below). 
	\begin{eqnarray*}
		\sum_{t\in T}\lambda_t+\sum_{i\in I} \theta_i &\leq &\alg,\\
		\sum_{t\in \Delta}\lambda_t +\theta_i &\geq & \alpha\, \min\{|\Delta|, B_i\}\quad \forall i\in I, \Delta\subseteq T,\\
		\lambda_t,\theta_i & \geq &0.
	\end{eqnarray*}
Comparing this with the constraints \eqref{dual1}---\eqref{dual3} that we introduce in Section \ref{sec:overview}, the main difference is that we only impose the second constraint for a specific set $\opt_i$, which is the set of arrivals matched to $i$ in the optimal offline solution. This difference does not play an important role here but can make a significant difference in settings with stochastic elements (see, for example, \citep{stochrew, full}).
 
	  These approaches do not yield a performance guarantee better than 0.5 in the general Adwords setting. In particular, \cite{Vazirani2022TowardsAP} identifies a key structural property, called \emph{no surpassing}, the absence of which prevents a generalization of their result to the Adwords setting. 
For the general setting, we replace the ``dominance" with a new lower bound (shown in Lemma \ref{alglbl}) and analyze a fractional algorithm ($\falg$) in order to analyze \alg. }
\section{Obstacles with Using Classic Primal-Dual Analysis}\label{appx:pmd}
The primal-dual framework of \cite{buchbind} and \cite{devanur} is a general technique for proving guarantees for online matching and related problems. To describe the framework, consider the following primal and dual problems that upper bound the optimal offline solution for the Adwords problem.

\begin{minipage}{7.5cm}
\begin{eqnarray}
\textbf{Primal:} &\min &\sum_{i\in I, t\in T}b_{i,t}\, x_{it} \nonumber\\
&s.t.\ & \sum_{t\in T}b_{i,t}x_{it} \leq B_i\quad \forall i\in I,\nonumber \\
&& \sum_{i\in I}x_{it} \leq 1\quad \forall t\in T,\nonumber \\
&& x_{it}\geq 0 \quad \forall i\in I, t\in T.\nonumber	
\end{eqnarray} 
\end{minipage}
\begin{minipage}{7.5cm}
\begin{eqnarray}
\textbf{Dual:}	&\min &\sum_{t\in T}\lambda_t +\sum_{i\in I} B_i\, \theta_{i} \nonumber\\
&s.t.\ & \lambda_{t} +b_{i,t}\theta_i \geq b_{i,t}\quad \forall i\in I, t\in T\nonumber \\
&& \lambda_t, \theta_{i}\geq 0 \quad \forall t\in T,i\in I.\nonumber	\\
&&\nonumber\\
&&\nonumber\\
&&\nonumber
\end{eqnarray} 
\end{minipage}
\smallskip

\noindent \textbf{Primal-dual certificate~\cite{devanur}:} 
To prove $\alpha$ competitiveness for $\alg$, 
it suffices to find a set of non-negative values $\lambda_t,\theta_i$ such that,
\begin{enumerate}[(i)]
\item$\lambda_t+b_{i,t}\theta_i\geq\alpha b_{i,t} \quad \forall i\in I, t\in T,$
\item $ \varepsilon\, \alg\geq \sum_{t\in T}\lambda_t + \sum_{i\in I} B_i\, \theta_i,$ 
\end{enumerate} 

To understand the obstacles in using this framework to analyze GPG for OBA, we start by defining a natural candidate solution for the system based on the decisions of \alg. 
Let \alg\ denote Algorithm \ref{rank} as well as its expected reward. In \alg, let $g(t)=e^{t-1},$ i.e., $\beta=1$. Given that we have external randomness in \alg\ through $Y=({y}_i)_{i\in I}$, we shall define variables $\lambda^{Y}_t$, $\theta^{Y}_i$ and subsequently set $\lambda^{{}}_t=E_{Y}[\lambda^{Y}_t]$ and $\theta^{{}}_i=E_{Y}[\theta^{Y}_i]$. Inspired by \cite{devanur}, we set $\lambda^{Y}_t$ and $\theta^{Y}_i$ as follows. 
Initialize all dual variables to 0. Conditioned on $Y$, for any match $(i,t)$ in \alg\ set, 
\begin{equation}
\lambda^{Y}_t=b_{i,t}(1-g(y_i)) \text{ and increment } \theta^{Y}_i \text{ by } \frac{b_{i,t}}{B_i}g(y_i).\label{dualset}
\end{equation}

Clearly, $\lambda^{Y}_t$ is set uniquely since \alg\ offers at most one resource $i$ to arrival $t$, and $\theta^{Y}_i$ takes a non-zero value only if it is also accepted by some $t$, and if this occurs $\theta^{Y}_i$ is never re-set. The following lemma (stated without proof) declares that this candidate solution satisfies constraint (ii) with $\varepsilon$ close to 1 in the small bid regime. 

%
\begin{lemma}
For the candidate solution given by \eqref{dualset}, constraint (ii) in the primal-dual certificate is satisfied with 
$\varepsilon=1+ \gamma$.
\end{lemma}

Unfortunately, there exist instances such that constraints (i) do not hold for any value of $\alpha>0.5$. For example, suppose that the first arrival has a non-zero bid exclusively from resource $i$. In fact, let $b_{i,1}=1$. Consequently, the first arrival is always matched to $i$ and we have, $\lambda^Y_1=  (1-g(y_i))\,\, \forall Y\in[0,1]^{n}$. Subsequent arrivals have higher bids from many resources other than $i$ such that they are matched to $i$ with very small probability. Thus, $\theta_i\approx \frac{1}{B_i}\,  \int_0^1g(x)\, dx$, and for $B_i\to+\infty$, we have 
\[E_{Y}[\lambda^Y_1]+\theta_i\approx   \left(\int_0^1 (1-g(x))\, dx + \frac{1}{B_i}  \int_0^1 g(x)\, dx \right)\to \int_0^1 (1-g(x))\, dx . \]
Notice that when $g(x)=e^{x-1}$, we have $\int_0^1 (1-g(x))\, dx =g(0)=1/e< 0.5$. 

Now, let $\opt$ denote the offline solution and let $\opt_i$ denote the set of arrivals matched to $i$ in $\opt$. In contrast to the primal-dual scheme, the LP free scheme 
imposes the following linear combination of the LP constraints,
\[\sum_{t\in\opt_i}\lambda_t + B_i\, \theta_i \geq \alpha \,  \opt_i,\]
where we summed constraints (i) over all arrivals in $\opt_i$ and scaled $\theta_i$ in accordance with the LP free system (the counterpart of constraint (ii) in the LP free system replaces $B_i\theta_i$ with $\theta_i$). Since $i$ is never matched to any arrival after the first one, we have that, $\lambda_t \geq b_{i,t}\,  (1-g(0))\,\, \forall t\geq 2$. Thus,
\[\sum_{t\in\opt_i}\lambda_t + \theta_i \geq  (1-1/e)\, \left(\sum_{t\in \opt_i,\, t\geq 2} b_{i,t}\right) +  \int_0^1 g(x)\, dx\, \geq\,  (1-1/e)\,\opt_i.\]
\section{Additional Examples to Demonstrate Main Analytical Challenge}\label{appx:challenge}
The first example shows that for \alg, even in the small bids regime, the total used budget of a resource $i$ can decrease as $y_i$ decreases and the bid prices of resource $i$ increase. The example also demonstrates that \alg\ and $\falg$ can output a very different matching on the same instance.
\begin{eg}\label{integfrac}
\emph{
	Consider an instance with $n$ resources $\{1,\cdots,n\}$, with 
	budget $n-1$ 
	for resources $j\in[n-1]$ and budget $(n-1)^{1.98}$ for resource $n$. 
	We focus on a snippet of this instance by considering arrivals $\{t+1,\cdots,t+2n-2\}\subset T$. We execute Algorithm \ref{rank} (with $\beta=1$) on this instance with seed $Y_{\mns n}$ for resources $j\in[n-1]$ fixed, and observe the change in output as seed $y_n$ varies. Suppose that exactly 1 unit of budget is available at arrival $t+1$ for every resource $j\in[n-1]$ and every value of $y_n\in[0,1]$. Further, $\forall y_n\in[0,1]$, suppose that all of resource $n$'s budget is available at $t+1$. The bids are as follows. 
\begin{enumerate}[(i)]
	\item Arrival $\tau\in\{t+1,\cdots,t+n-1\}$ bids $1$ for resource $j=\tau-t$, bids $2$ for resource $n$, and bids 0 for all other resources.  
	\item Arrival $\tau\in\{t+n,\cdots,t+2n-2\}$ bids $(n-1)^{0.99}$ for resource $j=\tau-(t+n-1)$, bids $(n-1)^{0.98}$ for resource $n$, and bids 0 for all other resources. 
\end{enumerate} 
Let $n\to +\infty$ so that we are in the small bid regime. Now, observe that there exists some $c\in(0,1)$ such that when $y_n\in(c,1]$, \alg\ matches resource $n$ to the last $n-1$ arrivals and all of $n$'s budget is utilized. On the other hand, for $y_n\in(0,c)$, \alg\ matches resource $n$ only to the first $n$ arrivals and almost none of $n$'s budget is used. Thus, as $y_n$ decreases from 1 to 0 and the bid price of resource $n$ increases, less of $n$'s budget is used in \alg. This is despite the fact that \alg\ makes allocation greedily w.r.t.\ the bid prices.} 
\end{eg}
$\falg$, which is not budget oblivious, does not exhibit the above behavior when started at arrival $t+1$ with the same initial conditions. Since at $t+1$ only 1 unit of budget is available for resources $j\in[n-1]$, for all values of $y_n$, $\falg$ will match at most a tiny fraction of the last $n$ arrivals to resources $j\in[n-1]$. In other words, for all values of $y_n$, $\falg$ (fractionally) matches resource $n$ to the last $n$ arrivals and uses up almost the entire budget of $n$. Thus, the matchings output by the two algorithms differ substantially for $y_n\in(0,c)$. 

Next, we consider an example that 
illustrates the challenge with analyzing $\falg$. As we discussed in Section \ref{sec:mainchal}, the main difficulty stems from the fact that the ``dominance" property does not hold. Consider a resource $i$ and fix all seed values except $y_i$. To define the example we use the notion of critical thresholds and the quantity $B(y_i)$ defined in Section \ref{sec:addnote}. Recall that $B(y_i)$ is the cumulative bid on $i$ from all arrivals matched to $i$ in \opt\ that have critical threshold at least $y_i$. 
The ``dominance" property states that 
the amount of $i$'s budget used is always at least $B(y_i)$. 

\begin{eg}\label{fracbid}\emph{
Consider a snippet of an instance with 3 resources $\{1,\cdots,3\}$, budget $n$ for resources 1 and 2 and $1.5n$ for resource 3. 
Consider a sequence of $[2n]$ arrivals. For every arrival $t\in[2n]$, we have $b_{i,t}=1$ from $i\in\{1,2\}$. Bids from resource $3$ are as follows, $b_{3t}=1$ for $t\in[n]$ and $b_{3t}=0.5$ otherwise. Notice that matching all $2n$ arrivals to resource 3 exaclty utilizes the budget of the resource.}

\emph{Let $n\to +\infty$. 
We execute $\falg$ with $\beta=1$ on this instance with seeds $(y_1,y_2)$ for resources $\{1,2\}$ fixed, and observe the change in output as seed $y_3$ varies. In particular, let $y_1\in(1/8,1/4)$ and $y_2\in(3/4,1)$. 
\begin{enumerate}[(i)]
\item For $y_3=1$, the first $n$ arrivals are matched to resource 1 and the next $n$ to resource 2. 
\item For $y_3\in (1/4, 0.4)$, on the last $n$ arrivals, the bid price of resource 3 is higher than the bid price of resource 2. Therefore, the first $n$ arrivals are matched to resource 1 and the next $n$ to resource 3.
\item For $y_3\leq 1/8$, the first $n$ arrivals are matched to resource 3. However, the last $n$ arrivals are now matched to resource 1.
\end{enumerate} 
Notice that the critical threshold of every arrival w.r.t.\ resource 3 is at least 1/8. Thus, \[B(y_3)=1.5n\,\, \forall y_3<1/8.\] 
Despite this, only the first $n$ arrivals are matched to resource 3 when $y_3\leq 1/8$ and the total budget of resource 3 used in $\falg$ does not exceed $n$ for $y_3<1/8$. }
\end{eg}
\section{Proof of Lemma \ref{lambda}}\label{appx:lambda}
\begin{repeatlemma}[Lemma \ref{lambda}.]
	Given $i\in I$ and seed $Y_{-i}$, for every  $y_i\in[0,1], t\in T,$ and $\tau\in [t,t+1)$, we have
	\[\lambda_\tau(y_i)\, \geq\, \lambda_{\tau}(1)\,\geq\,  b_{i,t}(1-g(y^c_{i,\tau})).\]
\end{repeatlemma}
\begin{proof}{Proof.}
	Given $i$ and $Y_{-i}$, consider an arbitrary seed $y_i \in[0,1]$ and arrival $t\in T$. By definition of $y^c_{i,\tau}$, we have
	$\lambda_\tau (1)\geq b_{i,t}\, (1-g(y^c_{i,\tau}))$.
	It remains to show that $\lambda_\tau(y_i)\geq \lambda_\tau(1)$. We claim that this follows from, 
	\[I(\tau,1)\backslash \{i\}\subseteq I(\tau,y_i).\]
	To see this, observe that given the above nesting of sets, we have
	\[\lambda_\tau(y_i)\,=\, \max_{j\in I(\tau,y_i)} b_{j,t}(1-g(y_j))\,\geq\, \max_{j\in I(\tau,1)\backslash\{i\}} b_{j,t}(1-g(y_j))\,=\, \max_{j\in I(\tau,1)} b_{j,t}(1-g(y_j))= \lambda_\tau(1).\]
	We prove that $I(\tau,1)\backslash\{i\}\subseteq I(\tau,y_i)$ by contradiction. Let $y_i$ be such that $ I(\tau, 1)\backslash\{i\} \not\subset I(\tau,y_i)$. Let $\tau_1\leq \tau$ be the earliest moment such that there exists a resource $i_1\in I(\tau_1,1)\backslash (I(\tau_1,y_i)\cup\{i\}),$ i.e., $i_1$ ($\neq i$) is available at $\tau_1$ in $\falg(1)$ but unavailable at $\tau_1$ in $\falg(y_i)$. This occurs only if $i_1$ is matched at some moment $\tau_0<\tau_1$ in $\falg(y_i)$ but not matched at $\tau_0$ in $\falg(1)$. Now, the following statements are true.
	\begin{enumerate}[(i)]
		\item $I(\tau_0,1)\backslash\{i\}\subseteq I(\tau_0,y_i)$. This follows from the definition of $\tau_1$ and the fact that $\tau_0<\tau_1$.
		\item $i_1\in I(\tau_0,1)\backslash\{i\}$. Follows from $i_1\in I(\tau_1,1)\backslash\{i\}$ and the fact that $I(\tau_1,1)\subseteq I(\tau_0,1)$.
	\end{enumerate}
	Overall, we have $i_1\in I(\tau_0,1)\backslash\{i\}\subseteq I(\tau_0,y_i)$.
	At every arrival, $\falg$ picks an available resource with highest bid price. So if 
	$\tau_0$ is matched to $i_1$ in $\falg(y_i)$, then from $(i)$ and $(ii)$ it must also be matched to $i_1$ in $\falg(1)$ (the bid price of $i_1$ at $\tau_0$ does not change). This contradicts the definition of $\tau_0$.
	\hfill\Halmos\end{proof}
\section{Obstacle to Analyzing \alg\ Directly}\label{appx:whyfrac} 
{\color{black}
Let $I(t,Y)$ denote the set of resources available at $t$ in \alg. Consider the candidate solution,
	\begin{eqnarray}
	&&\lambda_t = E_Y\left[\lambda_\tau(Y)\right]\quad \forall t\in T  \text{ and  }\quad
	\theta_i=E_Y[\theta_i(Y)]\quad  \forall i\in I,
	\label{gdualdef}
\end{eqnarray}
where,
\begin{eqnarray}
	&& \lambda_t(Y)= 
	\max_{j\in I(t,Y)} b_{j,t}(1-g(y_j))\qquad  \forall t\in T,\, Y\in [0,1]^n,\label{glambda1}\\
	&& \theta_i(Y)= \alg_i(Y)\, g(y_i)\qquad \forall i\in I,\, Y\in [0,1]^n.\label{gtheta1}
\end{eqnarray}
This definition is a natural generalization of the candidate solution used by \cite{devanur} to analyze \alg\ for vertex weighted online bipartite matching.

On every sample path $Y$, $\lambda_t(Y)>0\,\, \forall t\in T$ and $\theta_i(Y)>0\,\, \forall i\in I$. Thus, constraints \eqref{dual3} are  satisfied. We claim (without proof) that Lemmas \ref{dualreward} and \ref{lambda} from the analysis of $\falg$, can be generalized for \alg\ as stated below. 
\begin{lemma}
	For the candidate solution given by \eqref{dualdef}, we have,
	\[\sum_{i\in I} \theta_i+\sum_{t\in T} \lambda_t \leq \sum_{i\in I}\left(1+\max_{t\in T}\frac{b_{i,t}}{B_i}\right)
	\alg_i, \]
	and	constraint \eqref{dual1} 
	is satisfied with $\varepsilon=1+ \gamma$. 
\end{lemma}
\begin{lemma}
	Given $i\in I$ and seed $Y_{-i}$, for every  $y_i\in[0,1], t\in T$, we have
	\[\lambda_t(y_i)\, \geq\,  \lambda_{t}(1)\,\geq\,  b_{i,t}(1-g(y^c_i(t))).\]
\end{lemma}
Unfortunately, Lemma \ref{alglbl} need not hold for \alg\ and this is the main reason that we first analyze $\falg$ and then link the performance of $\alg$ with $\falg$. In the proof of Lemma \ref{alglbl}, the main step that fails is the one that links $\theta_i(y_i)$ to the difference $\lambda_{net}(y_i)-\lambda_{net}(1)$.  In \alg, we have that for every seed value $y\in[0,1]$ of resource $i$,
\[\lambda_{net}(y)\geq \alg_i(y)\,(1-g(y))+\sum_{j\in I\backslash\{i\}} \alg_j(y)\, (1-g(y_j)). \]
Notice that this is an inequality whereas in case of $\falg$ the left hand side equals the right hand side. In \alg, the value of $\lambda_{net}(y)$ may be strictly larger than the right hand side because \alg\ may overestimate the remaining budget of resources. In particular, when $i$ is matched to arrival $t$, the remaining budget of $i$ may be smaller than the bid $b_{i,t}$. However, we set $\lambda_t(y)=b_{i,t} (1-g(y))$, which may be  larger than the reward from matching $i$ to $t$. Now, consider the difference
\[\lambda_{net}(y_i)-\lambda_{net}(1)=(1-g(y_i))\sum_{t\in \falg_i(y_i)}b_{i,t}\,  +\sum_{j\in I\backslash\{i\}} \left[\sum_{t\in \falg_j(y_i)} b_{j,t}-\sum_{t\in \falg_j(1)}b_{j,t}\right]\, (1-g(y_j)), \]
where we used the fact that $g(1)=1$ for every value of $\beta>0$. In the analysis of $\falg$, we show that the second term is non-positive. However, in \alg\ the second term may be strictly positive. In particular, it is possible that for $j\in I\backslash\{i\}$, 
\[\alg_j(1)=\sum_{t\in \falg_j(1)}b_{j,t}=B_j\] 
but 
\[\sum_{t\in \falg_j(y_i)}b_{j,t}>B_j.\] 
In general, the increase in $\lambda_{net}(y_i)-\lambda_{net}(1)$ is not due to $\alg_i(y_i)$ alone and Lemma \ref{alglbl} does not hold for \alg.
}
\section{Proof of Theorem \ref{gpgub}}\label{appx:gpgub}
\begin{repeattheorem}[Theorem \ref{gpgub}.]
For every bid-to-budget ratio $\gamma\in(0,1]$ and $\beta\geq 1$, 
the competitive ratio of \alg\ is at most $0.624$.
\end{repeattheorem}
\begin{proof}{Proof.}
Consider a family of instances parameterized by $\epsilon\geq 0$, integer budget to bid ratio $m=\frac{1}{\gamma}\geq 1$, and $\beta\geq1$. Each instance in the family has $n+1$ resources. There are $2mn$ arrivals and each arrival bids $\frac{1}{m}$ on every resource in $\{1,2,\cdots,n\}$. Arrival $t\in[mn]$ bids \[w(t,\epsilon)=\frac{1-e^{\beta(\frac{t}{n+1}-1})}{m(1-e^{-\beta(1-\epsilon)})}\] 
on resource $n+1$. Arrivals $t\geq mn+1$ do not have an edge to resource $n+1$. Resources $\{1,2,\cdots,n\}$ have unit budget and resource $n+1$ has a budget of $\sum_{t\in[mn]} w(t,\epsilon)$. The optimal matching allocates resource $n+1$ to every $t\in[mn]$ and matches arrival $t$ to resource $\ceil{\frac{t}{m}}-n$ for $t\geq mn+1$. We have,
\[\opt=n+\sum_{t\in[mn]} w(t,\epsilon),\]
here $\lim_{n\to+\infty}\frac{1}{n}\sum_{t\in[mn]} w(t,0)\to \frac{e^\beta (1-\beta^{-1}) + \beta^{-1}}{e^\beta-1}$. Therefore, for $\epsilon=0$, we have, $\lim_{n\to+\infty}\frac{1}{n}\opt\to \frac{e^\beta (2-\beta^{-1})  -1+\beta^{-1}}{e^\beta-1}$. Let 
\[H(\beta)=\frac{e^\beta (2-\beta^{-1})  -1+\beta^{-1}}{e^\beta-1}.\] 
Observe that $H(1)=\frac{e}{e-1}$ and $H(\beta)>\frac{e}{e-1}$ for $\beta\geq1$. 

To understand the high level idea, consider $\epsilon=0$ and $\beta=1$. Let $y_i$ denote the random seed of resource $i\in[n+1]$ in \alg. 
For $k\in[n]$, let $z_k$ denote the $k$-th order statistic of the random variables $\{y_i\}_{i\in[n]},$ i.e., the $k$-th smallest seed value. For $k\in[n]$, random variable $z_k$ has mean $\frac{k}{n+1}$. Let $I(t)$ denote the (random) set of resources available at $t$ in \alg. 
Resource $n+1$ is matched to arrival $t\in[nm]$ if \[w(t,0)(1-e^{y_{n+1}-1})>\max_{i\in I(t)} \frac{1}{m}(1-e^{y_i-1})\geq \frac{1}{m}(1-g(z_t)),\]
here the last inequality follows from the fact that the $t$-th lowest seeded resource in $[n]$ will not be matched to any of the arrivals preceding $t$. For the sake of intuition, suppose that for $n\to+\infty$, we have $z_{\ceil{\frac{t}{m}}}\approx E[z_{\ceil{\frac{t}{m}}}]=\frac{{\ceil{\frac{t}{m}}}}{n+1}\,\, \forall t\in [mn]$, with high probability (w.h.p.). Then, resource $n+1$ is matched to $t$ only if $y_{n+1}\approx 0$, which occurs with a vanishingly small probability. By linearity of expectation, the expected total consumed capacity of resource $n+1$ is $\approx 0$. Therefore, the expected value of $\frac{1}{n}\alg$ is $\approx 1$, 
and we have $\frac{\alg}{\opt}\approx(1-1/e)$.


For the rest of this proof, we are interested in the asymptotic case where $n\to+\infty$. We use $O(\cdot)$ notation to hide constants independent of $n$. Now, observe that for a fixed $t$, $w(t,\epsilon)$ is a strictly increasing function of $\epsilon$. When we increase $\epsilon$ from 0, the value of offline (\opt) increases by 
\begin{equation}\label{optinc}
\sum_{t\in[mn]} (w(t,\epsilon)-w(t,0))=\sum_{t\in[mn]} w(t,0)\frac{e^{-\beta}(e^{\epsilon\beta}-1)}{1-e^{\beta(-1+\epsilon)}}
=n(H(\beta)-1)\frac{e^{-\beta}(e^{\epsilon\beta}-1)}{1-e^{\beta(-1+\epsilon)}}=O(\epsilon n).
\end{equation}
We show that the expected total value of \alg\ increases by an amount that is significantly lower than $O(\epsilon n)$. 
Let $z_0=0$ and $z_{n+1}=1$. The minimum of all ranks (first order statistic), $z_1$, is a random variable with mean $\frac{1}{n+1}$ and variance $\frac{n}{(n+1)^2(n+2)}=\frac{O(1)}{n^2}$~\citep{orderstat}. Using Chebyshev's inequality,
\[\mathbb{P}\left(\Big|\,z_1-z_0-\frac{1}{n+1}\,\Big|\geq \frac{1}{n^{4/3}}\right)\leq \frac{O(1)}{n^{4/3}}.\]
In general, for every $k\in\{0,1,\cdots,n\}$, the random variable $z_{k+1}-z_{k}$ has mean $\frac{1}{n+1}$ and variance $\frac{n}{(n+1)^2(n+2)}=\frac{O(1)}{n^2}$~\citep{orderstat}. Again, using Chebyshev's inequality,
\[\mathbb{P}\left(\Big|\,z_{k+1}-z_{k}-\frac{1}{n+1}\,\Big|\geq \frac{1}{n^{4/3}}\right)\leq \frac{O(1)}{n^{4/3}}\qquad \forall k\in[n]\cup\{0\}.\]
Applying the union bound, we have,
\[\mathbb{P}\left(\bigcup_{k\in[n]\cup\{0\}}\left\{\Big|\,z_{k+1}-z_{k}-\frac{1}{n+1}\,\Big|\geq \frac{1}{n^{4/3}}\right\}\right)\leq \frac{O(1)}{n^{1/3}}.\]
We have that for large $n$, the following event occurs w.h.p., 
\begin{equation}\label{goode}
\Big|\,z_{k+1}-z_{k}-\frac{1}{n+1}\,\Big|<\frac{1}{n^{4/3}}\qquad \forall k\in[n]\cup\{0\}.
\end{equation}
In the rest of this proof, we condition on this high probability event. Consequently, we have
\[\Big|z_{\ceil{\frac{t}{m}}}-\frac{\ceil{\frac{t}{m}}}{n+1}\Big|\leq \frac{\ceil{\frac{t}{m}}}{n^{4/3}}\qquad \forall t\in[mn].\]
We will show that in expectation, \alg\ uses a very small fraction of the budget of resource $n+1$. For any given $t\leq mn-mn^{3/4}$, resource $n+1$ is matched to arrival $t$ if,
\begin{eqnarray*}
m\,w(t,\epsilon)(1-g(y_{n+1}))&> &1-g(z_{\ceil{\frac{t}{m}}}),\\
\frac{1-e^{\beta\left(\frac{{\ceil{\frac{t}{m}}}}{n+1}-1\right)}}{1-e^{\beta(-1+\epsilon)}}(1-e^{\beta(y_{n+1}-1)})&>&1-e^{\beta\left(\frac{{\ceil{\frac{t}{m}}}}{n+1}+\frac{t}{n^{4/3}}-1\right)},\\
\Rightarrow\quad  y_{n+1}&< &\epsilon + \frac{O(1)}{n^{1/3}}.
\end{eqnarray*} 
For large enough $n$, resource $n+1$ is not matched to any arrival in \alg\ when $y_{n+1}>\epsilon + \frac{1}{n^{1/4}}$. In the following discussion, we condition on the event,
\begin{equation}\label{goode2}
y_{n+1} \leq \epsilon + \frac{1}{n^{1/4}}, 
\end{equation}
which occurs w.p.\ $\epsilon + \frac{1}{n^{1/4}}$ (independent of event \eqref{goode}). Finally, we claim that resource $n+1$ is matched to at most $\epsilon\, m(n+1) +mn^{3/4},$ i.e., $O(\epsilon mn)$, of the first $mn$ arrivals. To prove this, it suffices to show that for every $t\geq m(1+\epsilon(n+1)+n^{3/4})$, resource $\tau(t)={\ceil{\frac{t}{m}}}-n^{3/4}-\epsilon(n+1)$ is matched to one of the first $mn$ arrivals w.p.\ 1. Consider an arbitrary $t\geq m(1+\epsilon(n+1)+n^{3/4})$. Resource $\tau(t)$ is not matched to any arrival prior to $t$ if,
\begin{eqnarray*}
\frac{1-e^{\beta\left(\frac{{\ceil{\frac{t}{m}}}}{n+1}-1\right)}}{1-e^{\beta(-1+\epsilon)}}(1-e^{\beta(y_{n+1}-1)})\,&>&\, 1-e^{\beta\left(\frac{\tau(t)}{n+1}+\frac{{\ceil{\frac{t}{m}}}}{n^{4/3}}-1\right)},\\
&>&\, 1-e^{\beta\left(\frac{{\ceil{\frac{t}{m}}}}{n+1}-1-\epsilon\right)},
\end{eqnarray*}
which does not occur for any value of $y_{n+1}$. Overall, conditioned on events \eqref{goode} and \eqref{goode2}, the best that \alg\ can do is to match the first $O(\epsilon mn)$ arrivals to resource $n+1$, using up $O(\epsilon n)$ of the resource's budget. More precisely, w.p.\ 1, \alg\ uses at most,
\[n\frac{\int_0^\epsilon (1-e^{\beta(u-1)})du}{1-e^{\beta(-1+\epsilon)}} + o(n),\]
of resource $n+1$'s budget.
Unconditioning on events \eqref{goode} and \eqref{goode2}, the expected budget of resource $n+1$ used in \alg\ is $O(\epsilon^2 n)$. Thus,
\[\frac{\alg}{n}\leq 1+O(\epsilon^2)+o(1).\]
From \eqref{optinc}, we have
\[\frac{\opt}{n}\geq H(\beta)+O(\epsilon)+o(1).\]
Combining these inequalities, for $n\to+\infty$, we have $\frac{\alg}{\opt}\leq \frac{1}{H(\beta)}< (1-1/e)$ for some $\epsilon>0$ and any $\beta\geq 1$. As $m$ is arbitrary, the upper bound holds for small bids as well. 
We can obtain a stronger upper bound by considering a value of $\epsilon$ that minimizes the ratio in the limit $n\to+\infty,$ i.e.,
\[\min_{\epsilon\geq 0} \lim_{n\to+\infty}\frac{\alg}{\opt}\leq \min_{\epsilon\geq 0}\frac{1+\epsilon \frac{\int_0^\epsilon (1-e^{\beta(u-1)})du}{1-e^{\beta(-1+\epsilon)}}}{H(\beta)+(H(\beta)-1)\frac{e^{-\beta}(e^{\epsilon\beta}-1)}{(1-e^{\beta(-1+\epsilon)})}}<0.624\quad \forall \beta\geq 1,\]
where the final inequality follows by setting $\epsilon\approx 0.1336$.

\hfill\Halmos	
\end{proof}
\section{Proof of Theorem \ref{decomp}}
\label{appx:decompose} 
To prove Theorem \ref{decomp}, we first need the following intermediate result for $\falg$.
\begin{lemma}\label{lemdec}
Given $i\in I$, seed $Y_{-i}$ in $\falg$, and decomposable bids $b_{i,t}\in \{0,b_ib_t\}\,\, \forall i\in I, t\in T$, we have, 
\begin{eqnarray}
	&&E_{y_i}\left[\min\{\theta_i(y_i),g(y_i)\opt_i\}\mid Y_{-i}\right]+\sum_{t\in \opt_i} E_{y_i}[\lambda_t(y_i)\mid Y_{-i}] \, \geq\,  \sum_{v\in V}b(v)\Bigg[ 1-g(v)+
	\int_0^{v} g(x)dx  \Bigg].\nonumber
\end{eqnarray}
\end{lemma}
\begin{proof}{Proof.}
It suffices to show that,
\begin{equation}\label{thetalb}
	\theta_i(y_i) \geq  B(y_i)\, g(y_i)\quad \forall y_i\in[0,1).
\end{equation}
Given this inequality, using the fact that $\opt_i\geq B(y_i)\,\, \forall y_i\in[0,1]$ and taking expectation over $y_i\sim U[0,1]$, we have,
\begin{equation*}
	E_{y_i}\left[\min\{\theta_i(y_i),g(y_i)\opt_i\}\mid Y_{-i}\right] \geq \sum_{v\in V}\left(b(v)\,\int_0^{v} g(x)dx\right).
\end{equation*}
Combining the inequality above with the lower bound in Lemma \ref{lambda} gives us the desired.

Notice that, if $b_i=0$, then \eqref{thetalb} is trivially true. Further, all arrivals $t\in T$ where $b_t=0$ can be ignored. W.l.o.g., let $b_i>0$ and $b_{t}>0\,\, \forall t\in T$. Now, it suffices to show that 
\[\falg_i(y_i)\geq B(y_i)\,\,\forall y_i\in[0,1).\]
Given these inequalities, \eqref{thetalb} follows by definition of $\theta_i(y_i)$ (see \eqref{theta1}).  For the sake of contradiction, consider a value $y_i=y_0\in[0,1)$ such that $\falg_i(y_0)<B(y_0)\leq B_i$. Then, there exists an arrival $\tau\in \opt_i$ such that, $y_i^c(\tau)\geq y_0$, but $\tau$ is not matched to $i$ in $\falg(y_0)$. Since, $\falg(y_0)<B_i$, resource $i$ is available at $\tau$. Therefore, in $\falg(y_o)$, $\tau$ is matched to 
a resource $i_1$ such that, 
\[b_{i_1,\tau}(1-g(y_{i_1}))\,\geq\, b_{i,\tau}(1-g(y_0))\,\geq \,b_{i,\tau}(1-g(y^c_i(\tau))),\]
here the first inequality is strict w.p.\ 1 (ties occur w.p.\ 0). Using the decomposability of bids, we have (w.p.\ 1), 
\[b_{i_1}(1-g(y_{i_1}))>b_i(1-g(y_0)).\] 
Therefore, $i_1$ is preferred over $i$ at all arrivals. We say that $i_1$ is \emph{better} than $i$, or $i_1\succ i$ (in $\falg(y_0)$).  Now, notice that $i_1$ must be unavailable at $\tau$ in $\falg(1)$. Therefore, there exists an arrival $\tau_{i_1}$ prior to $\tau$ such that, $\tau_{i_1}$ is matched to $i_1$ in $\falg(1)$ but in $\falg(y_0)$, $\tau_{i_1}$ it is matched to a resource $i_2$ that is better than $i_1$. Repeating this argument a number of times, we get a sequence of arrivals $\tau>\tau_{i_1}>\cdots> \tau_{i_k}$ and resources $i\prec i_1\prec\cdots\prec i_k$. The number of resources is finite, so the sequence must terminate and, w.l.o.g., $i_k=i_\ell$ for some $\ell<k\leq n$, contradiction. 

\hfill\Halmos\end{proof}
\begin{repeattheorem}[Theorem \ref{decomp}.]
\alg\ with $\beta=1$ is $\frac{1}{1+\gamma}(1-1/e)$ competitive for OBA when the bids are decomposable, i.e., $b_{i,t}\in\{0,b_i\times b_t\}\,\, \forall i\in I, t\in T$. For the special case of OBA with integer starting capacities and binary bids, \alg\ is $(1-1/e)$ competitive for arbitrary bid-to-budget ratio $\gamma$.
\end{repeattheorem}
\begin{proof}{Proof.}
For $\falg$, when $g(x)=e^{x-1}$ we have $\forall i\in I,\, Y_{-i}\in[0,1]^{n-1}$,
\begin{eqnarray*}
	\sum_{t\in \opt_i} E_{y_i}[\lambda_t(y_i)\mid Y_{-i}]+ E_{y_i}[\theta_i(y_i)\mid Y_{-i}]
	&\geq & \sum_{v\in V}b(v)\left( 1-e^{-v}+
	e^{-v}-e^{-1} \right),\\
	&= &(1-e^{-1})\,\sum_{v\in V}b(v),\\
	&= &(1-e^{-1})\,\opt_i.
\end{eqnarray*} 
Now, combining this with Lemma \ref{dualreward} we have that $\falg$ with $\beta=1$ is $(1-1/e)$ competitive for decomposable bids. Using Theorem \ref{guafalg}, we have that \alg\ is $\frac{1}{1+\gamma} (1-1/e)$ competitive for decomposable bids.  The case of binary bids is a special case of decomposable bids where \alg\ and $\falg$ are identical. Thus, \alg\ with $\beta=1$ is unconditionally $(1-1/e)$ competitive for binary bids.

\hfill\Halmos\end{proof}
\section{Missing Proofs for Online Matching with Multi-Channel Traffic}\label{appx:multibase}
\noindent \textbf{Candidate solution for FRP:} Recall the candidate solution for FRP. For resource $i\in I$, let $o^\xtr_{i}$ be the total budget of $i$ allocated to external arrivals in \opt. Let $x^\xtr_i$ be the expected budget of $i$ allocated to external arrivals in $\falg$. Let $o^\ntr_i$ and $x^\ntr_i$ be the (expected) budget of $i$ allocated to internal arrivals in \opt\ and in $\falg$, respectively. Finally, let 
\[\zeta=\frac{\sum_{i\in I} (x^\ntr_i+x^\xtr_i)}{\sum_{i\in I}(o^\ntr_i+o^\xtr_i)}.\] 
Observe that $\zeta$ is equal to the performance ratio $\frac{\alg}{\opt}$. 
We start by noting that for $\delta>0$, all constraints (including \eqref{lpfree}) are well defined.  In particular, for $\delta>0$, using Lemma \ref{simplify}$(i)$ we have $\sum_{i\in I} o^\xtr_i >0$ on every non-trivial instance of the problem. By using the structure of the instance and borrowing ideas from \cite{multi}, it can be shown that the candidate solution satisfies constraints \eqref{cap}--\eqref{lx}. 
The crux of our analysis is to show that inequality \eqref{lpfree} is satisfied and we show this in the next section. The rest of the proof of Theorem \ref{crlb} is included in Appendix \ref{appx:multicr}.
\subsection{Proof of Inequality \eqref{lpfree}} \label{sec:multianalysis}
We start with some notation. Recall that $T^\ntr$ and $T^\xtr$ denote the set of internal arrivals and external arrivals, respectively. 
On sample path $Y$ in $\falg$, let $x^\xtr_i(Y)$ and $x^\ntr_i(Y)$ denote the budget of $i$ allocated to external and internal arrivals, respectively. Also, on sample path $Y$ in \alg, let $I(t,Y)$ denote the set of neighboring resources available at $t$. \alg\ may match more than $o^\ntr_i$ internal arrivals to $i$ on some sample paths. To capture these ``excess" matches we define, 
\[xs_i(Y)=\max\{x^\ntr_i(Y)-o^\ntr_i,\, 0\}\qquad \forall i\in I.\]
Excess matches play a crucial role in proving \eqref{lpfree}. In fact, using the LP free framework, we first establish a lower bound on $\zeta\sum_{i\in I} (o^\ntr_i+o^\xtr_i)$ in terms of the expectation $E_Y[\sum_{i\in I}xs_i(Y)\, g(y_i)]$. 
Given a candidate solution to the FRP, the random variables $\{xs_i(Y)\}_{i\in I}$ and their expected values are not uniquely defined. We address this by selecting the ``worst" possible way to define these random variables for any given candidate solution. 
To obtain a lower bound on $E_Y[\sum_{i\in I}xs_i(Y)\, g(y_i)]$ we now state two useful lemmas. The first lemma gives a lower bound on $xs_i(Y)$ in terms of $x^\xtr_i(Y)$, which also a random variable that is not uniquely defined (given a candidate solution). The second lemma gives a result for real valued functions that is reminiscent of the calculus of variations.
\begin{lemma} \label{xs}
	On every sample path $Y,$ $xs_i(Y)\geq o^\xtr_i-x^\xtr_i(Y)\quad \forall i\in I.$
\end{lemma}
\begin{proof}{Proof.}
	Suppose that every external arrival with an edge to $i$ is matched in $\falg(Y)$, then from Lemma \ref{simplify}$(i)$ we have, $x^\xtr_i(Y)= o^\xtr_i$, and the claim follows. If an external arrival with edge to $i$ is unmatched in $\falg(Y)$, then all of $i$'s budget is consumed in $\falg(Y),$ i.e., $x^\ntr_i(Y) +x^\xtr_i(Y)=B_i$. Using $B_i\geq o^\ntr_i+o^\xtr_i$ and the fact that $xs_i(Y)\geq x_i^\ntr(Y)-o^\ntr_i$,  completes the proof. 
	\hfill\Halmos\end{proof}
\begin{lemma} \label{func}
Given $g(x)=e^{x-1}$, $G(x)=\int_0^x g(u) du$, real values $h_0, H_0>0$ and an integrable function $h:[0,1]\to[0,H_0]$ with $\int_0^1 h(u) du=h_0$, we have \[\int_0^1 h(u)g(u)du\,\leq\, H_0\left(G(1)-G\left(1-\frac{h_0}{H_0}\right)\right).\]
\end{lemma}
\begin{proof}{Proof.}
	Consider the function 
	\[ f(u)=\begin{cases}
		0 \;\;\;\;\quad u\in[0,1-\frac{h_0}{H_0}),\\
		H_0 \quad u\in[1-\frac{h_0}{H_0},1].
	\end{cases} \]
	Observe that $\int_0^1 f(u) du=h_0$ and $\int_0^1 f(u)g(u)du=H_0\left(G(1)-G(1-\frac{h_0}{H_0})\right)$.
	It suffices to show that, $\int_0^1 h(u)g(u)du\leq \int_0^1 f(u)g(u)du$. Let $\mathcal{H}$ denote the set of all integrable functions $h:[0,1]\to [0,H_0]$ with $\int_0^1h(u)du=h_0$. Let \[h^*=  \sup_{h\in \mathcal{H}}\int_0^1 h(u)g(u) du.\]
	Now, if $\int_0^1f(u)g(u)du < \int_0^1h^*(u)g(u)du $, then there exists a closed interval $I_1\subseteq [0,1-\frac{h_0}{H_0})$ with non-zero measure and a closed interval $I_2\subseteq[1-\frac{h_0}{H_0},1]$ with non-zero measure such that, 
	\[h^*(u)>0\,\, \forall u\in I_1\quad \text{  and   }\quad h^*(u)<H_0\,\, \forall u\in I_2.\] 
	Let $\epsilon=\min\left\{\min_{u\in I_i} h^*(u),\min_{u\in I_2}H_0-h^*(u)\right\}$. Let $|I|$ denote the length/measure of an interval $I$. Consider function $h'$ such that $h'(u)=h^*(u)-\epsilon|I_1|^{-1}\,\, \forall u\in I_1$, $h'(u)=h^*(u)+\epsilon|I_2|^{-1}\,\, \forall u\in I_2$ and $h'(u)=h^*(u)$ for all other values of $u\in[0,1]$. We have,
	$\int_0^1 h'(u) du=h_0$ and $\int_0^1 h'(u)g(u)du>\int_0^1h^*(u)g(u)du$, contradiction.  
	\hfill\Halmos\end{proof}
Next, we state two intermediate inequalities that we use to establish inequality \eqref{lpfree}.
\begin{lemma}\label{plpfree}
On every problem instance under consideration, the candidate solution to FRP satisfies, 
\begin{enumerate}[(i)]
\item 	$\zeta\, \sum_{i\in I} (o^\ntr_i+o^\xtr_i) \geq\, \sum_{i\in I}\left(G(1)\, o^\ntr_i + x^\xtr_i+ E_Y[xs_i(Y)\,g(y_i)]\right).$ 
\item  $\sum_{i\in I} E_{Y}[ xs_i(Y)\,g(y_i)]\,\geq\,  \sum_{i\in I} o^\xtr_i    G\left(1-\frac{\underset{j\in I}{\sum}x^\xtr_i}{\underset{j\in I}{\sum}o^\xtr_i}\right).$
\end{enumerate}

\end{lemma}
\noindent Before we prove Lemma \ref{plpfree}, observe that a straightforward combination of inequalities $(i)$ and $(ii)$ in the lemma proves that the candidate solution satisfies \eqref{lpfree}.
\begin{proof}{Proof of Lemma \ref{plpfree}.}
The proof of part $(i)$ is based on the LP free framework. In particular, we find a feasible solution to the following system of inequalities,
\begin{eqnarray}
\sum_{t\in T} \lambda_t + \sum_{i\in I} \theta_i &\leq & \zeta\, \sum_{i\in I} (o^\ntr_i+o^\xtr_i)\label{mdual1}\\
\theta_i+\sum_{t\in \opt_i} \lambda_t  &\geq & G(1)o^\ntr_i+x^\xtr_i+E_Y[xs_i(Y)\,g(y_i)] \qquad \forall i\in I, \label{mdual2}\\
\lambda_t\geq 0, &&\theta_i\geq 0\qquad \forall t\in T,\, i\in I. \label{mdual3}
\end{eqnarray}
Given a feasible solution the system, summing up inequalities \eqref{mdual2} over all $i\in I$ gives us the desired. 
Our candidate feasible solution is based on the random variables,
\begin{eqnarray}
&& \theta^\ntr_i(Y)= x^\ntr_i(Y)\,	 g(y_i)\qquad \forall i\in I,\label{thintr}\\
&& \theta^\xtr_i(Y)= x^\xtr_i(Y) \qquad \forall i\in I,\label{thextr}\\
&&	 \lambda_t(Y)= 
\begin{cases}
	\max_{j\in I(t,Y)} (1-g(y_j))\qquad  &t\in T^\ntr,\\
	0 \qquad  &t\in T^\xtr.
\end{cases} 
\end{eqnarray}
We note that random variables $\{\lambda_t(Y)\}_{t\in T^\ntr}$ and $\{\theta^\ntr_i(Y)\}_{i\in I}$ are the same as their counterparts (\eqref{lambda1} and \eqref{theta1}) in Section \ref{sec:ana}. For external arrivals, the fact that there is exactly one edge incident on each arrival eliminates the need to 
set a non-zero value for $\lambda_t$ and we set $\lambda_t(Y)=0\,\, \forall \in T^\xtr$.  Consider the (non-negative) candidate solution,
\begin{eqnarray*}
&&\lambda_t = E_Y\left[\lambda_t(Y)\right]\quad \forall t\in T, \qquad 	\theta_i=E_Y[\theta^\ntr_i(Y)+\theta^\xtr_i(Y)]\quad  \forall i\in I.
\end{eqnarray*}
Observe that,
\begin{eqnarray*}\label{summer}
E_{Y}\left[\sum_{i\in I} \left(\theta^\ntr_i(Y)+\theta^\xtr_i(Y)\right)+\sum_{t\in T}\lambda_t(Y)\right]=\sum_{i\in I} (x^\ntr_i+x^\xtr_i),
\end{eqnarray*}
Thus, our candidate solution satisfies inequality \eqref{mdual1}.
It remains to show that inequality \eqref{mdual2} holds for all $i\in I$. We prove this via the same high level approach as the analysis of $\falg$ in Section \ref{sec:main}. Fix an arbitrary resource $i\in I$ and seeds $Y_{-i}$ for all resources except $i$. To prove \eqref{mdual2}, it suffices to show that,
\begin{equation}\label{mdual4}
E_{y_i}\left[\theta^\ntr_i(Y)+\theta^\xtr_i(Y)+\sum_{t\in \opt_i} \lambda_t(Y)\mid Y_{-i}\right] \geq\, G(1)\, o^\ntr_i + x^\xtr_i +E_{y_i}[xs_i(Y)\,g(y_i)\mid Y_{-i}]. 	
\end{equation}	
%
To show inequality \eqref{mdual4}, we first consider 
a ``truncated" instance of the problem without the external arrivals, i.e., the arrival sequence is $T^\ntr$ and all budgets and bids are unchanged. Let $\opt^\ntr$ denote the optimal offline solution of the truncated instance. 
{\color{black} W.l.o.g., $\opt^\ntr$ is identical to the internal traffic portion of \opt. This follows from the fact that \opt\ matches every arrival in the original instance (Lemma \ref{simplify}$(i)$) and $\opt^\ntr$ can do no better than match all arrivals. } 
For every $i\in I$, exactly $o^\ntr_i$ of resource $i$'s budget is used in $\opt^\ntr$ and the set of arrivals matched to $i$ is given by $\opt^\ntr_i= T^\ntr\cap \opt_i$.
Now, using Lemma \ref{lemdec} for decomposable bids (see Appendix \ref{appx:decompose}), we have 
\begin{equation}\label{borrow}
	E_{y_i}\left[\min\{x^\ntr_i(Y),\,o^\ntr_i\}\,g(y_i)+\sum_{t\in \opt^\ntr_i} \lambda_t(Y) \mid Y_{-i} \right]\, \geq\, G(1)\, o^\ntr_i,
\end{equation}
where, $x^\ntr_i(Y)\, g(y_i)=\theta_i^\ntr(Y)$, by definition (see \eqref{thintr}).
Now, we claim that,
\[\theta^\ntr_i(y_i)+\theta^\xtr_i(y_i)=(\min\{x^\ntr_i(y_i),o^\ntr_i\}+xs_i(y_i))g(y_i)+x^\xtr_i(y_i)\quad \forall y_i\in[0,1].\]
Combining this equality with \eqref{borrow} gives us \eqref{mdual4} and completes the proof of part $(i)$ of the lemma. To show the equality above, fix  an arbitrary $y_i\in[0,1]$ and consider the following cases. 


\noindent \textbf{Case I}: $x^\ntr_i(y_i)\geq  o^\ntr_i$. We have, $o^\ntr_i+xs_i(y_i)=x^\ntr_i(y_i)$.
From \eqref{thintr} and \eqref{thextr},
\[\theta^\ntr_i(y_i)+\theta^\xtr_i(y_i)=(o^\ntr_i+xs_i(y_i))g(y_i)+x^\xtr_i(y_i),\]
as desired.

\smallskip

\noindent \textbf{Case II:} $x^\ntr_i(y_i)< o^\ntr_i$ ($\leq B_i$), i.e., resource $i$ is available at every internal arrival in $\falg(y_i)$. 
In this case, $xs_i(y_i)=0$ and
\[\theta^\ntr_i(y_i)+\theta^\xtr_i(y_i)=(x^\ntr_i(y_i)+xs_i(y_i))g(y_i)+x^\xtr_i(y_i),\]
as desired.

Part $(ii)$ of the lemma gives a lower bound on $ \sum_{i\in I} E[g(y_i)\, xs_i(y_i)]$ in terms of $x^\xtr_i$ and  $o^\xtr_i$. Since $xs_i(y_i)< o^\xtr_i$, we can ignore any resource $i\in I$ for which $o^\xtr_i=0$ (for this part). Let $o^\xtr_i>0\,\, \forall i\in I$ and recall that $G(x)=\int_0^x g(u)du=e^{x-1}-e^{-1}$. Now, to prove $(ii)$ we claim that it suffices to show, 
\begin{equation}\label{dlb2}
	E_{y_i}[ xs_i(y_i)\,g(y_i)]\,\geq\,  G\left(1-\frac{x^\xtr_i}{o^\xtr_i}\right)o^\xtr_i\qquad \forall i\in I, Y_{-i}\in[0,1]^{n-1}.
\end{equation} 
First, let us show why \eqref{dlb2} is sufficient. 
\begin{eqnarray}
	\sum_{i\in I} G\left(1-\frac{x^\xtr_i}{o^\xtr_i}\right)o^\xtr_i &= &\left(\sum_{i\in I} o^\xtr_i \right)\left[\sum_{i\in I}  G\left(1-\frac{x^\xtr_i}{o^\xtr_i}\right)\frac{o^\xtr_i}{\sum_{j\in I} o^\xtr_i}\right],\nonumber\\
	&\geq &\left(\sum_{i\in I} o^\xtr_i \right) \left[  G\left(1-\frac{\sum_{j\in I}x^\xtr_i}{\sum_{j\in I}o^\xtr_i}\right)\right],\label{jensen}
\end{eqnarray}
here inequality \eqref{jensen} follows by using the fact that $G(x)$ is a convex function of $x$ and applying Jensen's inequality.  It remains to show \eqref{dlb2}. First, by a direct application of Lemma \ref{xs},
\[E_{y_i}[xs_i(y_i)\, g(y_i)]\geq E_{y_i}[(o^\xtr_i-x^\xtr_i(y_i))\, g(y_i)]=\, G(1)\,o^\xtr_i-\int_0^1 x^\xtr_i(u)g(u)du.\]
Using the fact that there are exactly $o^\xtr_i$ external arrivals with an edge to $i$ (Lemma \ref{simplify}$(i)$), we have, $x^\xtr_i(y_i)\leq o^\xtr_i$. Further, $x^\xtr_i(u)$ is an integrable function of $u$ with $\int_0^1 x^\xtr_i(u) du=x^\xtr_i$. 
Applying Lemma \ref{func} with $h(u)=x^\xtr_i(u)$, $h_0=x^\xtr_i$, and $H_0=o^\xtr_i$, we have, 
\[\int_0^1 x^\xtr_i(u)g(u)du\,\leq\, o^\xtr_i\left(G(1)-G\left(1-\frac{x^\xtr_i}{o^\xtr_i}\right)\right)\]
Therefore,
\[E_{y_i}[xs_i(y_i)\, g(y_i)]\geq \, G(1)\,o^\xtr_i-o^\xtr_i\left(G(1)-G\left(1-\frac{x^\xtr_i}{o^\xtr_i}\right)\right)\geq\, G\left(1-\frac{x^\xtr_i}{o^\xtr_i}\right)o^\xtr_i.\]

\hfill\Halmos\end{proof}

\subsection{Proof of Theorem \ref{crlb}}\label{appx:multicr}
\begin{repeattheorem}[Theorem \ref{crlb}.]
For any given $\delta\in(0,1]$, the optimal value of FRP is a lower bound on the competitive ratio of $\falg$. For $\delta\geq \frac{1}{e}$, the optimal value of FRP is at least $1+\delta \ln \delta$.
\end{repeattheorem}
\begin{proof}{Proof.}
Given an instance with fraction of external traffic $\delta$, we show that there exists a feasible solution to FRP with objective value equal to the performance ratio $\frac{\alg}{\opt}$ on that instance. Consider the following candidate solution for the FRP. 
For resource $i\in I$, let $o^\xtr_{i}$ be the total budget of $i$ allocated to external arrivals in \opt. Let $x^\xtr_i$ be the expected budget of $i$ allocated to external arrivals in $\falg$. Let $o^\ntr_i$ and $x^\ntr_i$ be the (expected) budget of $i$ allocated to internal arrivals in \opt\ and in $\falg$, respectively. Finally, let 
\[\zeta=\frac{\sum_{i\in I} (x^\ntr_i+x^\xtr_i)}{\sum_{i\in I}(o^\ntr_i+o^\xtr_i)}.\] 
Observe that $\zeta$ is equal to the performance ratio $\frac{\pg}{\opt}$. We start by noting that for $\delta>0$, all constraints (including \eqref{lpfree}) are well defined.  In particular, for $\delta>0$, using Lemma \ref{simplify}$(i)$ we have $\sum_{i\in I} o^\xtr_i >0$ on every non-trivial instance of the problem.
Now, constraints \eqref{cap} follow from the fact that \opt\ cannot match a resource more times than the resource's initial budget. 
Since \opt\ matches every arrival (Lemma \ref{simplify}$(i)$), no algorithm can match more than $o^\xtr_i$ external arrivals to resource $i\in I$. Thus, the candidate solution satisfies \eqref{capx}. 
Constraint \eqref{fet} follows directly from the definition of $\delta$. 
Constraint \eqref{lx} is inspired by a constraint in the factor revealing program of \cite{multi}. To show that the candidate solution satisfies \eqref{lx}, we start by giving a lower bound on the expected number of internal arrivals matched by \alg. Using the fact that 
\opt\ matches all arrivals we have,
\[		\sum_{i\in I} x^\ntr_i\leq \sum_{i\in I} o^\ntr_i\leq  (1-\delta)\sum_{i\in I} B_i.\]
Applying the inequalities above in the definition of $\zeta$, we get,
\begin{eqnarray*}
\sum_{i\in I} x^\xtr_i&= &\zeta\sum_i (o^\ntr_i+o^\xtr_i)-\sum_{i\in I} x^\ntr_i,\\
& \geq &\zeta\, \delta\sum_{i\in I} B_i - (1-\zeta)\sum_{i\in I} o^\ntr_i,\\
&\geq & (\zeta\, \delta -(1-\zeta)(1-\delta))\sum_{i\in I} B_i,\\
&=& (\zeta-1+\delta)\sum_{i\in I} B_i.
\end{eqnarray*}
Recall that we showed constraint \eqref{lpfree} is satisfied by the candidate solution in Section \ref{sec:multianalysis}. It remains to show that the optimal value of FRP is at least $1+\delta\ln \delta$ for $\delta\geq 1/e$.
Let $W^\xtr= \sum_{i\in I} o^\xtr_i$ and $X^\xtr=\sum_{i\in I}x^\xtr_i$. From \eqref{lx}, we have, \begin{equation}\label{lxmod}
X^\xtr\,\geq\, (\zeta-1+\delta)\left(\sum_{i\in I} B_i\right)\,\geq 0,
\end{equation}
here the non-negativity of the expression follows from the fact that $\zeta\geq (1-1/e)\geq 1-\delta$. 
From \eqref{lpfree}, we have
\begin{eqnarray*}
\zeta \sum_{i\in I} (o^\ntr_i+o^\xtr_i)&\geq &\eta\,\sum_{i\in I} o^\ntr_i + W^\xtr \left(\frac{X^\xtr}{W^\xtr}+   G\left(1-\frac{X^\xtr}{W^\xtr}\right)\right),\\
\zeta \sum_{i\in I} o^\ntr_i +\zeta\,\delta \sum_{i\in I} B_i	&\overset{(a)}{\geq} &\eta\,\sum_{i\in I} o^\ntr_i + \delta \left(\sum_{i\in I} B_i\right) \left(\frac{\zeta-1+\delta}{\delta}+   G\left(1-\frac{\zeta-1+\delta}{\delta}\right)\right),\\
(\zeta-\eta)\frac{\sum_{i\in I} o^\ntr_i}{\sum_{i\in I} B_i} +\zeta\,\delta &\overset{}{\geq} & \delta\left(\frac{\zeta-1+\delta}{\delta}+   G\left(\frac{1-\zeta}{\delta}\right)\right)  ,\\
(\zeta-\eta) (1-\delta) +\zeta\, \delta&\overset{(b)}{\geq} &\zeta-1+\delta+ \delta  G\left(\frac{1-\zeta}{\delta}\right),\\
(1-\eta)\frac{1-\delta}{\delta}&\overset{(c)}{\geq}  & e^{\frac{1-\zeta-\delta}{\delta}}-e^{-1},\\
\zeta&\geq &1-\delta\left(1+\ln\left(\frac{1}{e\delta}\right)\right),\\
&= &1+\delta\ln \delta.
\end{eqnarray*}
Inequality $(a)$ follows from \eqref{fet}, \eqref{lxmod}, and the fact that $J(u)=u+G(1-u)$ is an increasing function for $u\geq 0$. Inequality $(b)$ follows from the upper bound $\sum_{i\in I} o^\ntr_i \leq (1-\delta)\sum_{i\in I} B_i$ that is implied by \eqref{cap} and \eqref{fet}. Finally, inequality $(c)$ follows by plugging in $G(u)=e^{u-1}-e^{-1}$.
\hfill\Halmos\end{proof}
\subsection{Other Missing Proofs for Multi-Channel Traffic}\label{appx:multi}
\begin{repeatlemma} [Lemma \ref{simplify}.]
	For any given $\delta\in[0,1]$,	to lower bound the competitive ratio of \pg\ it suffices to consider instances where 
	\begin{enumerate}[(i)]
		\item \opt\ matches every arrival. 
		\item Internal traffic arrives before 
		external traffic, i.e., there is no internal traffic after the first external arrival.
	\end{enumerate}
\end{repeatlemma}
\begin{proof}{Proof.}
	To prove part $(i)$, it suffices to show that for any graph $G$ and capacities $(B_i)_{i\in I}$, removing unmatched arrivals from $G$ does not decrease the ratio $\frac{\alg}{\opt}$. 
	Clearly, removing arrivals that are not matched in \opt\ has no effect on \opt. It suffices to show that 
	on every sample path $Y,$ $\pg(Y)$ does not increase if we remove an arrival $t\in T$. Consider two cases based on the type of resource matched to $t$ in $\pg(Y)$. First, suppose that $t$ is either unmatched or matched to a resource with non-zero remaining budget at the end of the arrival sequence. Then, removing $t$ does not change the matching. Second, suppose that $t$ is matched to resource $i$ such that $\pg_i(Y)=B_i$. Removing $t$ may change the matching but does not increase the total budget consumption of any resource.
	
	Consider an external arrival $t\in T$ in an arbitrary arrival sequence. First, notice that changing the sequence of arrivals does not change \opt. To prove part $(ii)$, it suffices to show that moving arrival $t$ to the end of the arrival sequence (keeping all other arrivals fixed) can not increase $\pg$. Consider an arbitrary sample path $Y$ in $\pg$. If $t$ is unmatched in $\pg(Y)$, moving $t$ to the end of the sequence does not change the matching. However, if $t$ is matched in $\pg(Y)$ then moving $t$ to the end of the arrival sequence does not increase the total budget consumption of any resource. 
	\hfill\Halmos\end{proof}
\section{Missing Proofs for Online Matching with Stochastic Rewards}\label{appx:stoch}
One common way to state asymptotic competitive ratio results for vanishing probabilities is as follows,
\begin{equation}\label{10}
\textsc{ALG}\geq\, \alpha\, \opt - o(1),
\end{equation}
where $\alpha$ is the competitive ratio and $o(1)$ is a term that goes to 0 as $\max_{i\in I, t\in T} p_{i,t} \to 0$. This is the form used (explicitly) in \cite{huang} and (implicitly) in \cite{mehta}. An alternative approach is to define an entirely multiplicative ratio as follows,
\[\textsc{ALG}\geq (1-o(1))\,\alpha\, \opt,\]
where $o(1)$ is a term that goes to zero as $\max_{i\in I, t\in T} p_{i,t} \to 0$.
Guarantees of this form were shown in \cite{deb} and \cite{stochrew}. For simplicity, we state and prove our result in the sense of \eqref{10}. 

\noindent \textbf{From stochastic rewards to Adwords:} Given an instance of online matching with stochastic rewards where the edge probabilities are given by $p_{i,t}$ and $p=\max_{i\in I, t\in T} p_{i,t}$, consider the Adwords instance on the same graph with bids \[b_{i,t}=p_{i,t}\quad \forall i\in I, t\in T.\]
and unknown stochastic budgets,
\[B_i\sim Exp(1) \quad \forall i\in I,\]
here $Exp(1)$ is the exponential distribution with mean 1. Taking expectation over the budgets, we have the following performance guarantee for this stochastic instance of Adwords.

\begin{lemma}\label{stochads}
For the stochastic instance of Adwords described above, \alg\ is at least $0.522-o(1)$ competitive for $\beta=1.15$ and at least $0.508-o(1)$ competitive for $\beta=1$ (in expectation over the random budgets and seed values).
\end{lemma}
\begin{proof}{Proof.}
Given starting budgets $\mathcal{B}=\{B_i\}_{i\in I}$, let 
\[\gamma_i(B_i)=\max_{t\in T} \frac{b_{i,t}}{B_i}\leq \frac{p}{B_i}\quad \forall i\in I.\]
Let $\opt(\mathcal{B})$ and $\alg(\mathcal{B})$ denote the (expected) total reward in \opt\ and \alg\ respectively. For any fixed value of $\beta$, let $\alpha$ denote the asymptotic performance guarantee of \alg\ in the limit $\gamma\to+\infty$. We have,
\[\alpha(\mathcal{B})=\frac{\alg(\mathcal{B})}{\opt(\mathcal{B})}\geq \frac{1}{1+\max_{i\in I} \gamma_i(B_i)} \alpha\geq \frac{\min_{i\in I} B_i}{\min_{i\in I} B_i+p}\alpha.\]
Thus, $E_{\mathcal{B}}[\alpha(\mathcal{B})]\geq E_{\mathcal{B}}[\frac{\min_{i\in I} B_i}{\min_{i\in I} B_i+p}]\, \alpha$. Using the inequality $\frac{1}{1+x}\geq 1-x$, we have,
\[E_{\mathcal{B}}\left[\frac{\min_{i\in I} B_i}{\min_{i\in I} B_i+p}\right]\geq 1-E_{\mathcal{B}}\left[\frac{p}{\min_{i\in I} B_i}\right].\]
For any finite set $I$, let $E$ denote the event that $\min_{i\in I} B_i \geq \sqrt{p}$. Notice that the probability that $E$ occurs equals 1 in the limit of $p\to 0$. Thus, $\alpha > E_{\mathcal{B}}[\alpha(\mathcal{B})]\geq \alpha-o(1)$, as desired.
\hfill\Halmos\end{proof}
\begin{repeatlemma}{Lemma 14 in \cite{huang}.}\label{restate}
An online algorithm that is $\alpha-o(1)$ competitive for the stochastic instances of Adwords defined above, is also $\alpha-o(1)$ competitive
for online matching with stochastic rewards.
\end{repeatlemma}
\begin{proof}{Proof of Theorem \ref{stoch}.}
Using Lemma 14 in \cite{huang} (stated above) and Lemma \ref{stochads}, gives us the desired.
\hfill\Halmos	\end{proof}
\section{Potential Application in Automated Budget Management}\label{appx:budmage}
In using a automated tools to manage their portfolio, the advertiser decides the overall budget, creates a portfolio of ad campaigns, and specifies a high level performance goal for the portfolio. Using these specifications, the tool aims to automatically determine (over time) a good budget distribution for the portfolio, as well as, daily budgets and bids for key words in each campaign. As the example below illustrates, fixing a (daily) budget at the start of the day may be highly sub-optimal.
\begin{eg}
	\emph {Consider an advertiser with a portfolio composed of two search ad campaigns labeled $\{1,2\}$. In the absence of a budget constraint, let the total expenditure (per day) in campaign 1 be a Bernoulli random variable $X_1\in\{0,1\}$. Similarly, let $X_2\in\{0,1\}$ represent the expenditure (per day) in campaign 2. We know that $X_1$ and $X_2$ are identically distributed with mean 0.5 and $X_1=1-X_2$, i.e., the campaign expenditures are in perfect negative correlation. Now, suppose we have a total daily budget of 1 that needs to be distributed between the two campaigns at the start of each day.  If we distribute the overall budget evenly i.e., allocate a daily budget of 0.5 to each campaign, then the total budget utilization in campaign $i\in\{1,2\}$ is at most $E[\min\{0.5, X_i\}]=0.25$, i.e., only $50$\% of the total budget is used. Splitting the budget across the campaigns in a different proportion does not improve the total utilization. 
	}
	
	\emph{ 
		Now, suppose that instead of fixing the daily budget for each campaign at the beginning, we start by giving a budget of 0.5 to each campaign and use live cross-campaign data to adjust the distribution during the day. Since the campaigns are negatively correlated, exactly one of the two campaigns will have a non-zero expenditure on each day. Let $i^*$ denote this campaign on a given day. 
		Before the initial budget of $i^*$ runs out, we can redirect the rest of the budget to $i^*$ and achieve a budget utilization of 100\%. 
}\end{eg}
\smallskip

Platforms with in-house tools for automated budget optimization (such as Search Ads 360 by Google) may, in fact, have the cross-campaign data necessary to perform adjustments to budget distribution in real-time. At a conceptual level, armed with a budget oblivious allocation algorithm and live cross-campaign data, such a platform can start the day with some distribution of budget to each campaign and adjust the distribution during the day to improve overall utilization. When a campaign is about to run out of its tentatively assigned budget, the platform can either transfer unused budget from another active campaign to this one, or, let the campaign expire for the day. In this way, campaigns that do particularly well during the day would receive more budget and this improves the overall utilization of budget for the portfolio. A budget oblivious allocation algorithm is agnostic to such adjustments and only needs to be notified when the budget of a campaign expires for the day. 

\end{APPENDICES}
\end{document}